\documentclass{article}
\usepackage{graphicx,amsmath,amsthm, amssymb,array,psfrag}
\RequirePackage{bbm}
\newtheorem{theorem}{{Theorem}}

\newtheorem{lemma}{{Lemma}}
\newtheorem{corollary}{\textbf{Corollary}}

\allowdisplaybreaks
\newcommand {\bE} {\mathbb{E}}
\newcommand {\bP} {\mathbb{P}}
\newcommand {\magn} {m_1}
\newcommand {\mbr} {\underline{r}}

\newcommand {\mbn} {\underline{n}}
\newcommand {\mbN} {\underline{N}}
\newcommand {\mbx} {\underline{x}}
\newcommand {\mbw} {\underline{w}}
\newcommand {\mbs} {\underline{s}}
\newcommand {\mby} {\underline{y}}
\newcommand {\mbz} {\underline{z}}
\newcommand {\mbY} {\underline{Y}}
\newcommand {\mbu} {\underline{u}}
\newcommand {\mbv} {\underline{v}}
\newcommand {\mbU} {\underline{U}}
\newcommand {\mbX} {\underline{X}}
\newcommand {\mcZ} {{\underline{{\cal Z}}}}
\newcommand {\mbm} {\underline{m}}

\newcommand {\mbh} {\underline{h}}

\newcommand {\mbone} {\underline{1}}
\newcommand {\mts} {\textbf{s}}
\newcommand {\mtr} {\textbf{r}}
\newcommand {\mtv} {\textbf{v}}
\newcommand {\mtS} {\textbf{S}}

\newcommand {\mtR} {\textbf{R}}
\newcommand {\mtV} {\textbf{V}}
\newcommand {\mtt} {\textbf{t}}
\newcommand{\ind}{\mathbbm{1}}
\newcommand{\indicator}[1]{\ind_{\{ #1 \}}}

\newcommand {\blangle} {\Big\langle}
\newcommand {\brangle} {\Big\rangle}
\newcommand {\free}{\mathcal{F}}

\newcommand {\naturals}{\mathbb{N}}
\newcommand {\mcuZ}{\underline{\mathcal{Z}}}
\title{\LARGE \bf Tight Bounds on the Capacity of Binary Input random CDMA Systems}

\author{{Satish Babu Korada and  Nicolas Macris}\\
\\
School of Information and Communication Sciences\\
Ecole Polytechnique F\'ed\'erale de Lausanne\\
LTHC-IC-Station 14, CH-1015 Lausanne\\
 Switzerland}

\begin{document}
\maketitle

\begin{abstract}
\noindent We consider multiple access communication on a binary input additive white Gaussian
noise channel using randomly spread code division. For a general class of
symmetric distributions for spreading coefficients, in the limit of a large
number of users,
we prove an upper
bound on the capacity,  which matches a formula that Tanaka obtained by using the replica method. 
We also
show concentration of various relevant quantities including mutual information,  capacity and free energy.
The mathematical methods are quite general and allow us to discuss extensions to other multiuser
scenarios. 
\end{abstract} 

\section{Introduction}
Code Division Multiple Access (CDMA) has been a successful scheme for
reliable communication between multiple users and a common receiver. The scheme consists of $K$ users modulating their information
sequence by a signature sequence, also known as spreading sequence, of length
$N$ and transmitting. The number $N$ is sometimes referred to as the spreading gain or the
number of chips per sequence. The receiver obtains the sum of all transmitted signals and
the noise which is assumed to be white and Gaussian (AWGN).

The achievable rate region (for real valued inputs) with
 power constraints and optimal decoding has been given in \cite{Ver86}.
There it is shown that the achievable rates depend only on the correlation
matrix of the spreading coefficients. It is well known that these detectors have
exponential (in $K$) complexity. Therefore, it is important to analyze the
performance under sub-optimal but low-complexity detectors like the linear
detectors. For a good overview of these detectors we refer to \cite{Verd98}. 
In \cite{VeS99}, the authors considered random spreading (spreading sequences
are chosen randomly) and analyzed the spectral efficiency, defined as the
bits per chip that can be reliably transmitted, for these detectors.
In the {\em large-system limit} $(K\to \infty, N \to \infty, \frac{K}{N} = \beta)$ they
obtained nice analytical formulas for the spectral efficiency and showed that
it concentrates. These formulas
follow from the known spectrum of large covariance matrices. In \cite{TsZ00},\cite{TsH99} the authors analyzed the signal
to interference ratio for the decorrelator and the MMSE receiver and showed that
it is asymptotically Gaussian with variance going to zero. 


Now consider the case where the user input is restricted to take
only binary values. Not much is known in this case except for the spectral
efficiency in the case of high SNR which is analyzed in \cite{TsV00}. The random
matrix techniques used for Gaussian inputs do not apply here because the
spectral efficiency cannot be written in terms of just the covariance matrix of
the spreading sequences.
Tanaka \cite{Tan02} applied the formal replica method, developed in statistical
mechanics, to this problem and conjectured the formula for spectral efficiency
and bit error rate (BER) for uncoded transmission.
These results were later extended in \cite{GuV05} to include the case of unequal powers and
channel with fading. The replica method is non-rigorous but  believed to yield exact results for some models in statistical mechanics \cite{Tal03}. More recently
Montanari and Tse \cite{MoT06itw} have made progress towards a rigorous
derivation of Tanaka's capacity formula in a restricted range of parameters.

Our main contributions in this paper are twofold. First we prove that Tanaka's
formula is an upper bound to the capacity for all values of the parameters
and second we prove various useful concentration theorems in the large-system limit.

\subsection{Statistical Mechanics Approach}

There is a natural connection between various communication systems and statistical
mechanics of random spin systems, stemming from the fact that often in both systems
there is a large number of degrees of freedom (bits or spins), interacting
locally, in a random environment. So far, there have
been applications of two important but somewhat complementary approaches of
statistical mechanics of random systems.

The first one is the very important but mathematically
uncontrolled replica method. The merit of this
approach is to obtain conjectural but rather explicit formulas for quantities of
interest such as, free energy, conditional entropy or error probability. In some
cases the natural fixed point structure embodied in the mean field formulas
allows to guess good iterative algorithms. This program has been carried out for
linear error correcting codes, source coding, multiuser settings like broadcast channel
(see for example \cite{Mon01}, \cite{KaT05}, \cite{NKMS03}) and the case of
interest here \cite{Tan02}:  randomly spread CDMA  with binary
inputs. 

The second type of approach aims at a rigorous understanding of the replica
formulas and has its origins in methods stemming from mathematical physics (see
\cite{GuT02,GuT03}, \cite{Tal03}). For systems whose underlying degrees of
freedom have Gaussian distribution (Gaussian input symbols or Gaussian spins in
continuous spin systems) random matrix methods can successfully be employed.
However when the degrees of freedom are binary (binary information symbols or
Ising spins) these  seem to fail, but  the recently developed interpolation
method \cite{GuT02},\cite{GuT03} has had some success\footnote{Let us point out
that, as will be shown later in this paper, the interpolation method can also
serve as an alternative to random matrix theory for Gaussian inputs.}.  The
basic idea of the interpolation method is to study a measure which interpolates
between the posterior measure of the ideal decoder and a  mean field measure.
The later can be guessed from the replica formulas and from this perspective the
replica method is a valuable tool. So far this program has been developed only for
linear error correcting codes on sparse graphs and binary input  symmetric
channels \cite{Mon05}, \cite{KuM06}. 

In this paper we develop the interpolation method for
 the random CDMA system with binary inputs (in the large-system limit). The
situation is qualitatively different than the ones mentioned above in that the ``underlying graph'' is complete. Superficially one might think that it is similar to  the Sherrington-Kirkpatrick model which was the first one treated by the interpolation method. However as we will see the analysis of the randomly spread
CDMA system is substantially different due to the structure of the interaction between degrees of freedom.

\subsection{Communication Setup}\label{sec:communicationsetup}
We consider a scenario where $K$ users send binary information symbols
$\mbx=(x_1,\dots,x_K)^t$, $x_k \in \{\pm 1\}$ to a common receiver, through a single AWGN
channel. Each user $k$ has a random signature
sequence $\underline{s}_k = (s_{1k},...,s_{Nk})^t$ where the components are independently identically distributed.
For each time division (or chip) interval $i=1,...,N$ the received signal $\mby=(y_1,...,y_N)$ is
\begin{equation*}
y_i=\frac{1}{\sqrt N}\sum_{k=1}^Ks_{ik} x_k +  \sigma n_i
\end{equation*}
where ${\mbn} = (n_1,...,n_N)^t$ are independent identically distributed
Gaussian  variables $\mathcal N(0,1)$ so that the noise power is $\sigma^2$. The variance of $s_{ik}$ is set to $1$ and the scaling factor $1/\sqrt N$ is introduced so that the power (per symbol)
of each user is normalized to 1. 
Our results hold for the rather wide class of distributions satisfying:
\vskip 0.25cm
\noindent{\it {\bf Assumption A.} The distribution $p(s_{ik})$ is symmetric 
$$
p(s_{ik})=p(-s_{ik})
$$ 
and has a rapidly decaying tail. More precisely, 
%
there exists positive constants $s_0$ and $A$ such that $\forall s \geq s_0$
$$
p(s_{ik} \geq s)\leq e^{-A s^2}
$$

}
In particular, our favorite Gaussian and binary cases are included in this class, and also any compactly supported distribution. An inspection of our proofs suggests that the results could be extended to a larger class satisfying:
\vskip 0.25cm
\noindent{\it {\bf Assumption B.} The distribution $p(s_{ik})$ is symmetric with finite second and fourth moments.}
\vskip 0.25cm
\noindent However to keep the proofs as simple as possible only one of the theorems is proven with such generality.

In the sequel we use the notations $\mts$ for the $N\times K$ matrix $(s_{ik})$, $\mtS$ for the corresponding random matrix, and $\mbX$, $\mbY$ for the input and output random vectors.

Our main interest is in proving a ``tight" upper bound on
\begin{equation}\label{capacity}
C_K=\frac{1}{K}\max_{p_{\mbX}}\bE_{\mtS}[I(\mbX;\mbY)]
\end{equation}
in the large-system limit $K\to +\infty$ with $\frac{K}{N}=\beta$ fixed. 
In the next few paragraphs we discuss various settings for which it is justified to consider this formula as a capacity. In principle for multiaccess channels one maximizes over product distributions $p_{\mbX}(\mbx)=\prod_{k=1}^Kp_k(x_k)$. But in fact this restriction makes no difference when one maximizes the expected mutual information because the maximum is attained for a uniform distribution.
Indeed for any given $\mts$ the mutual information $I(\mbX;\mbY)$ is a concave functional of $p_{\mbX}$ and thus so is its average. Moreover the later is invariant under the transformations
$p_{\mbX}(x_1,x_2,...,x_K)\to p_{\mbX}(\epsilon_1x_1, \epsilon_2x_2,...,\epsilon_Kx_K)$ where $\epsilon_i=\pm 1$. Combining these two facts we deduce that the maximum in \eqref{capacity} is attained for the convex 
combination 
\begin{equation*}
\frac{1}{2^K}\sum_{\epsilon_1,...,\epsilon_K}
p_{\mbX}(\epsilon_1x_1,...,\epsilon_Kx_K) = \frac{1}{2^K}
\end{equation*}
which is nothing else than the product of uniform distributions for each user.
Before discussing the meaning of \eqref{capacity} for the CDMA setting let us note that it can also be interpreted as the capacity of a MIMO system with binary constellations, $K$ transmit, $N$ receive antennas, and ergodic channel coefficients
$s_{ik}$ that are known to the receiver only \cite{Tel99}, \cite{FoG98}.


In the traditional CDMA setting (see for example \cite{Verd98}) the spreading sequences are assigned to each user and do not change from symbol to symbol. Moreover it is assumed that the users and the receiver know $\mts$. The general analysis of multiaccess channels implies that the total capacity per user (or maximal achievable sum rate) is
\begin{equation}\label{usualcapacity}
\frac{1}{K}\max_{\prod_{k=1}^{K}p_k(x_k)} I(\mbX;\mbY)
\end{equation}
where the maximum is over $p_i(x)=p_i\delta(x-1)+(1-p_i)\delta(x+1)$ and $p_i\in [0,1]$, $i=1,...,k$. In the large-system limit we are able to prove a concentration theorem for the mutual information $I(\mbX;\mbY)$ which implies that if $(p_1,...,p_K)$ belongs to a finite discrete set ${\cal D}$ with cardinality increasing at most polynomially in $K$, then \eqref{usualcapacity} concentrates on 
$\frac{1}{K}\max_{p\in {\cal D}}\bE_{\mtS}[I(\mbX;\mbY)]$. Of course by the same
argument as before this maximum is attained 
for $p=\frac{1}{2}$ as long as $\frac{1}{2}\in {\cal D}$. 
Unfortunately, in order to extend these arguments to the more realistic case of
exponential cardinality of ${\cal D}$, or even  all possible continuous values of the input distribution (and thus to fully justify \eqref{capacity}) we would have to prove stronger forms of concentration.

At this point it is interesting to discuss the situation for the 
continuous input case. There it is known that the maximum of
\eqref{usualcapacity} is attained for a Gaussian input distribution independent
of the spreading sequence realization \cite{Ver86}. Then the
concentration theorems for $I(\mbX;\mbY)$ suffice to prove that in the large-
system limit \eqref{usualcapacity} asymptotically equals \eqref{capacity}.
It is an open problem to decide if an analogous result holds in the binary input case, namely that the maximum of \eqref{usualcapacity} is attained for the uniform distribution. We conjecture that this is the case.

Alternatively, following \cite{VeS99} one may consider the case of ``long
spreading sequences", that is sequences that extend over many symbol durations.
Then by ``ergodicity" one can compute the capacity as an expectation of
\eqref{usualcapacity} over $\mtS$. In the continuous input case it turns out
that one can switch
the expectation and the maximum because it can be shown (by the standard
argument adapted above for the binary case) that the maximum of the expectation
is attained for the same Gaussian input distribution. Thus, remarkably,  in the
continuous case one exchanges the expectation over $\mtS$ with the maximum over
product distributions even for finite $K$. 
 
Finally let us return to the binary case and consider the situation of long
spreading sequences as in \cite{VeS99} that are assumed to be unknown (or rather
not used) to the encoder and known to the receiver. Then, by the analysis in
\cite{Tel99}, formula \eqref{capacity} gives the capacity. If users do not cooperate $p_{\mbX}$ is really a product distribution. But in any case the maximum is attained for the uniform distribution. 
 
Let us now collect a few formulas that will be useful in the sequel.
The conditional entropy $H(\mbX\mid\mbY)=\bE_{{\mbY}\vert \mts}[H(\mbX\mid\mby)]$
is the average 
over $\mbY$ given $\mts$ of the Shannon entropy for the posterior distribution
\begin{equation}\label{eqn:postdist}
p(\mbx\mid\mby,\mts) = \frac{p_{{\mbX}}(\mbx)}{Z(\mby,\mts)}\exp\bigl(-\frac{1}{2\sigma^2}\Vert \mby-N^{-{\frac{1}{2}}}\mts\mbx\Vert^2\bigr) 
\end{equation}
with the normalization factor 
\begin{equation}\label{eqn:partition}
Z(\mby,\mts)=\sum_{\mbx}p_{{\mbX}}(\mbx)e^{-\frac{1}{2\sigma^2}\Vert \mby-N^{-{\frac{1}{2}}}\mts\mbx\Vert^2}
\end{equation}
Note that this is the distribution used by the ideal or optimal detector.
The average over $\mbY$ is carried out with 
the distribution induced by the channel transition probability
\begin{align}\label{ydistribution}
 p(\mby\mid \mts)=\sum_{\mbx^0}p_X(\mbx^0) \frac{e^{-\frac{1}{2\sigma^2}\Vert\mby- N^{-{\frac{1}{2}}}\mts\mbx^0\Vert^2}}{(\sqrt{2\pi}\sigma)^{N}} =\frac{1}{(\sqrt{2\pi}\sigma)^N}Z(\mby,\mts)
\end{align}
where in the sum $\mbx^0$ is interpreted as the input signal. 
The normalization factor \eqref{eqn:partition} can be interpreted as the partition function of interacting Ising spins $x_k=\pm 1$ with free measure $p_{\mbX}$. In view of this it is not surprising that the free energy
\begin{equation}\label{eqn:freeenergy}
 f(\mby,\mts)=\frac{1}{K}\ln Z(\mby,\mts)
\end{equation}
plays a crucial role. In appendix \ref{apen:capfree} we show that it is related to the mutual information by
\begin{equation}\label{eqn:entrofreerel}
 \frac{1}{K}I(\mbX;\mbY)= -\frac{1}{2\beta} - \bE_{\mbY\mid\mts}[f(\mby,\mts)]
\end{equation}
Therefore 
\begin{equation}\label{eqn:capfreerel}
C_K=-\frac{1}{2\beta} -\min_{p_{\mbX}}\bE_{\mbY,\mtS}[f(\mby,\mts)]
\end{equation}
Of course by the previous discussion the $\min_{p_{\underline X}}$ is attained for $p_{\underline X}(\underline x)=\frac{1}{2^K}$. 


\subsection{Tanaka's formula for binary inputs}
By using the formal replica trick of statistical mechanics Tanaka reduced the calculation of the conditional entropy to a variational problem. His conjectural formula is
\begin{align}\label{tanaka}
\lim_{K\to\infty}C_K = \min _{m\in[0,1]} c_{RS}(m)
\end{align}
where the ``replica symmetric capacity functional"
\begin{align} \label{eqn:hrs}
c_{RS}(m)  =  \frac\lambda2 (1+m)
 - \frac{1}{2\beta}\ln\lambda\sigma^2 -\int Dz \ln(2 \cosh(\sqrt
\lambda z + \lambda))
\end{align}
with
\begin{align}\label{lambda}
\lambda &= \frac{1}{\sigma^2 + \beta(1-m)}
\end{align}
and $Dz$ the standard Gaussian measure
$Dz \equiv \frac{e^{-\frac{z^2}{2}}}{\sqrt {2\pi}}dz$, has to be maximized over
a parameter\footnote{this parameter can be interpreted as the expected value of
the MMSE estimate for the information bits} $m$. It is easy \footnote{using
integration by parts formula for Gaussian random variables} to see that the
maximizer must satisfy the fixed point condition
\begin{align}\label{eqn:fixedpt}
m = \int Dz \tanh(\sqrt \lambda z + \lambda) 
\end{align}
The formal calculations involved in the replica method make clear that the
formula \eqref{tanaka} should not depend on the distribution of the spreading
sequence (see \cite{Tan02}). 

In the present problem one expects a priori that replica symmetry is not broken because of a gauge
symmetry induced by channel symmetry. For this reason Tanaka's formula is
conjectured to be exact. Our upper bound (Theorem~\ref{thm:upperbound}) on the
capacity precisely coincides with the above formulas and strongly supports this conjecture.

Recent work announced by Montanari and Tse \cite{MoT06itw} also provides strong support to the conjecture at least in a regime of $\beta$ without phase transitions (more precisely, for $\beta \leq \beta_s(\sigma)$ where $\beta_s(\sigma)$ is the maximal value of $\beta$ such that the solution of (\ref{eqn:fixedpt}) remains unique). The authors first solve the case of sparse signature sequence (using the area theorem and the data processing inequality) in the limit $K\to\infty$. Then
the dense signature sequence (which is of interest here) is recovered by
exchanging the $K\to\infty$ and $sparse\to dense$ limits.

\subsection{Gaussian inputs}

In the case of continuous inputs $x_k\in \mathbb{R}$,
in formulas \eqref{eqn:partition}, \eqref{ydistribution} $\sum_{\mbx}$ are replaced by 
$\int d\mbx$. The capacity is maximized by a Gaussian prior,
\begin{align}\label{eqn:gaussianprior}
p_{{\mbX}}(\mbx)=\frac{e^{-\frac{\vert\vert\mbx\vert\vert^2}{2}}}{(2\pi)^{N/2}}
\end{align}
and one can express it in terms of a determinant involving the correlation matrix of the spreading sequences. Using the 
exact spectral measure given by random matrix theory Shamai and Verdu
\cite{VeS99} obtained  the rigorous result 
\begin{align}\label{eqn:gaussianverdu}
\lim_{K\to\infty}C_K =&
\frac{1}{2}\log(1+\sigma^{-2}-\frac{1}{4}Q(\sigma^{-2},\beta))\nonumber\\&+\frac{1}{2\beta}\log(1+\sigma^{-2}\beta-\frac{1}{4}Q(\sigma^{-2},\beta))-
\frac{Q(\sigma^{-2},\beta)}{8\beta\sigma^{-2}}
\end{align}
where
$$
Q(x,z)=
\left(\sqrt{x(1+\sqrt{z})^2+1}-\sqrt{x(1-\sqrt{z})^2+1}\right)^2
$$
On the other hand Tanaka applied the formal replica method to this case and found \eqref{tanaka} with
\begin{align}\label{eqn:gaussianreplica}
c_{RS}(m) = \frac{1}{2}\log(1+\lambda)
-\frac1{2\beta}\log\lambda\sigma^2-\frac{\lambda}{2}(1-m)
\end{align}
where $\lambda = (\sigma^2+\beta(1-m))^{-1}$. The maximizer satisfies
\begin{align}\label{eqn:saddlegaussian}
m=\frac{\lambda}{1+\lambda}
\end{align}
Solving \eqref{eqn:saddlegaussian} we obtain  $m=\frac{\sigma^2}{4\beta}Q(\sigma^{-2},\beta)$ and substituting this in (\ref{eqn:gaussianreplica})
gives the equality between (\ref{eqn:gaussianverdu}) and
(\ref{eqn:gaussianreplica}). So at least for the case of Gaussian inputs we are already
assured that the replica method finds the correct solution.

As we will show in section \ref{sec:gaussianinter} our methods also work in the case of Gaussian inputs,
and yield the upper bound. 

\subsection{Contributions and organization of this work}

The main focus and challenge of this work is on the case of binary inputs for
the communication set up described above, although the methods also work for many other
constellations including Gaussian inputs. The main results
are explained in section \ref{mainresults} while the remaining sections are devoted to the proofs.

We prove concentration of the mutual information in the limit of
$K\to+\infty$ and $\beta=\frac{K}{N}$ fixed (Theorems~\ref{thm:capconc},
\ref{thm:capconcbin} in section \ref{sec:concentration}). As we will see the
mathematical underpinning of this is the concentration of a more fundamental
object, namely, the ``free energy" of the associated spin system
(Theorem~\ref{thm:freeconc}). In fact this turns out to be important in the
proof of the bound on capacity. When the spreading coefficients are Gaussian the
main tool used is a powerful theorem \cite{Tal03} of the concentration of Lipschitz functions of many independent
Gaussian variables, and this leads to subexponential concentration bounds. For
more general spreading coefficient distributions such tools do not suffice and we
have to combine them with martingale arguments which lead to weaker algebraic
bounds. Since the concentration proofs are mainly technical they are presented
in appendices \ref{apen:probtools}, \ref{apen:concproofs}.

Sections \ref{sec:upperboundproof} and \ref{sec:concmag} form the core of the
paper. They detail the proof of the main
Theorem~\ref{thm:upperbound} announced in section \ref{sec:upperbound}, namely
the tight upper bound on capacity. We use ideas from the interpolation method combined
with a non-trivial concentration theorem for the empirical average of soft bit
estimates. 

Section~\ref{sec:indspreadingproof} shows that the average capacity is independent of
the spreading sequence distribution at least for the case where it is symmetric
and decays fast enough (Theorem~\ref{thm:indspreading} in section \ref{sec:indspreading}). This enables us to restrict
ourselves to the case of Gaussian spreading sequences which is more amenable to
analysis. The existence of the limit $K\to\infty$ for the capacity is shown in
section \ref{sec:limitexistproof}.

Section ~\ref{sec:extensions} discusses various extensions of this work. We
sketch the treatment for unequal powers for each user as well as colored noise.
As alluded to before the bound on capacity for the case of Gaussian inputs can
also be obtained by the present method and we give some indications to this
effect.

The appendices contain the proofs of various
technical calculations. Preliminary versions of the results obtained in this paper have been
summarized in references \cite{KoM07isit} and \cite{KoM07allerton}.

\section{Main Results}\label{mainresults}

\subsection{Concentration}\label{sec:concentration}

In the case of a  Gaussian input signal, the concentration can be deduced from
general theorems on the concentration of the spectral density for  random
matrices, but this approach breaks down for binary inputs. Here we prove,
\begin{theorem}[{\bf concentration of capacity, Gaussian spreading sequence, binary inputs}]\label{thm:capconc}
\noindent Assume the distribution $p(s_{ik})$ are standard Gaussians. Given
$\epsilon>0$, there  exists an integer $K_1= O(\vert\ln\epsilon\vert)$
independent of $p_{\mbX}$, such that for all $K  > K_1$, 
\begin{equation*}
\bP[\vert I(\mbX;\mbY) - \bE_\mtS [I(\mbX;\mbY)]\vert\geq \epsilon K ] \leq 3 e^{-\alpha_1 \epsilon^2 K}
\end{equation*}
where $\alpha_1=\frac{1}{16}\sigma^{4}(64\beta+32+\sigma^2)^{-1}$.
\end{theorem}
\vskip 0.25cm

The mathematical underpinning of this result is in fact a more general
concentration result for the free energy \eqref{eqn:freeenergy}, that will be of some use latter on.
\begin{theorem}[{\bf concentration of free energy, Gaussian spreading sequence, binary inputs.}]\label{thm:freeconc}
\noindent Assume the distribution $p(s_{ik})$ are standard Gaussians. Given $\epsilon>0$, there exists an integer $K_2= O(\vert\ln\epsilon\vert)$ independent of $p_{\mbX}$, such that 
for all $K \geq K_2$,
\begin{equation*}
\bP[\vert f(\mby,\mts) - \bE_{\mbY,\mtS} [f(\mby, \mts)]\vert\geq \epsilon ] \leq 3 e^{-\alpha_2 \epsilon^2 \sqrt K}
\end{equation*}
where $\alpha_2=\frac{1}{32}\sigma^4\beta^{\frac{3}{2}}(2\sqrt\beta+\sigma)^{-2}$.
\end{theorem}
\vskip 0.25cm
\noindent We prove these theorems thanks to powerful probabilistic tools
developed by Ledoux and Talagrand for Lipschitz functions of many Gaussian random variables.
These tools are briefly reviewed in Appendix \ref{apen:probtools} for the
convenience of the reader and the proofs of the theorems are presented in
Appendix \ref{apen:concproofs}.
Unfortunately the same tools do not apply directly to the case of other
spreading sequences. However in this case the following weaker result can at
least be obtained.
\begin{theorem}[{\bf concentration, general spreading
sequence}]\label{thm:capconcbin}  Assume the spreading sequence satisfies
assumption B. There exists
an integer $K_1$ independent of $p_{\mbX}$, such that for all $K > K_1$
\begin{align*}
\bP[|I(\mbX;\mbY) - \bE_{\mtS}[I(\mbX;\mbY)]|\geq \epsilon K] \leq \frac{\alpha}{K
\epsilon^2}
\end{align*}
\begin{align*}
\bP[f(\mby,\mts) - \bE_{\mbY,\mtS}[f(\mby,\mts)]|\geq \epsilon] \leq
\frac{\alpha}{K
\epsilon^2}
\end{align*}
for some constant $\alpha > 0$ and independent of $K$. 
\end{theorem}
To prove such estimates it is enough (by Chebycheff) to control second moments.
For the mutual information we simply have to adapt martingale arguments of
Pastur, Scherbina and Tirrozzi, \cite{PaS91,ShT93} whereas the case of free energy is more complicated because of the additional Gaussian noise fluctuations. We deal with these by combining martingale arguments and Lipschitz function techniques.

The concentration of capacity, namely
\begin{equation}\label{capcon}
 \bP[\vert
\max_{p_{\mbX}}I(\mbX;\mbY)-\max_{p_{\mbX}}\bE_{\mtS}[I(\mbX;\mbY)]\vert\geq
\epsilon K]\leq \frac{\alpha}{K\epsilon^2}
\end{equation}
would follow from a stronger (uniform concentration with respect to $p_{\mbX}$)
\begin{align}\label{stronger}
\bP[\max_{p_{\mbX}}|I(\mbX;\mbY) - \bE_{\mtS}[I(\mbX;\mbY)]|\geq \epsilon K]
\leq \frac{\alpha}{K\epsilon^2}
\end{align}
To see this it suffices to note that for two positive functions $f$ and $g$ we have $\vert \max f -\max g\vert\leq \max\vert f -g\vert$. But unfortunately it is not clear how to extend our proofs to obtain \eqref{stronger}. However as announced in the introduction we can deduce \eqref{stronger} from our theorems, by using the union bound, as long as the maximum is carried out over a finite set (sufficiently small with respect to $K$) of distributions.

We wish to argue here that Theorem \ref{thm:freeconc} suggests a method for
proving the concentration of the bit error rate (BER) for uncoded communication
\begin{equation}
 \frac{1}{2} (1 -  \frac{1}{K}\sum_{k=1}^{K} x_{0,k}\hat x_{k})
\end{equation}
where
 the MAP bit estimate for uncoded communication is defined through the marginal
of (\ref{eqn:postdist}), namely
$\hat{x}_k = \text{argmax}_{x_k=\{\pm 1\}}p(x_k\mid\mby,\mts)$.
We remark that 
$$
\hat x_k={\rm sign}\langle x_k\rangle
$$ 
where we find it convenient to adopt the statistical mechanics notation 
$\langle - \rangle$ 
for the average with respect to the posterior measure (\ref{eqn:postdist}). For example the average 
$$
\langle x_k \rangle=\sum_{\mbx} x_k p(\mbx\mid \mby,\mts)
$$ 
(a soft bit estimate or ``magnetization") can be obtained from the free energy by adding first an infinitesimal 
perturbation (``small external magnetic field") to the exponent in (\ref{eqn:postdist}), namely
$h\sum_{k=1}^{K} x_k^0 x_k$, and then differentiating the perturbed free energy\footnote{we do not write explicitly the $h$ dependence in the perturbed free energy},
\begin{equation*}
\frac{1}{K}\sum_{k=1}^{K} x_k^0\langle x_k\rangle=\lim_{h\to 0}\frac{d}{dh}\frac{1}{K}\ln Z(\mby,\mts)
\end{equation*}
However one really needs to relate 
${\rm sign}\langle x_k\rangle$ to the derivative of the free energy and this does not appear to be obvious.
One way out is to
 introduce product measures of $n$ copies (also called ``real replicas") of the posterior measure 
$$
p(\mbx^{(1)} \mid \mby,\mts)p(\mbx^{(2)} \mid \mby,\mts)
\ldots p(\mbx^{(n)} \mid \mby,\mts)
$$
and then relate
$$
\sum_{k=1}^{K} (x_k^0 \langle x_k\rangle)^n=
\sum_{k=1}^{K} \langle x_k^0 x_k^{1}...x_k^0 x_k^{n}\rangle_n 
$$
to a suitable derivative of the replicated free energy. Then from the set of all
moments one can in principle reconstruct ${\rm sign}\langle x_k\rangle$. Thus
one could try to deduce the concentration of the BER from the one for the free
energy. However the completion of this program requires a  uniform, with respect
the system size, control of the derivative of the free energy precisely at
$h=0$, which at the moment is still lacking\footnote{however this can be done for
Lebesgue almost every $h$}.

\subsection{Independence with respect to the distribution of the spreading
sequence}\label{sec:indspreading}

The replica method leads to the same Tanaka formula for general class of
symmetric distributions $p(s_{ik})=p(-s_{ik})$. We are able to prove this: in
particular binary and Gaussian spreading sequences lead to the same capacity.

\begin{theorem}\label{thm:indspreading}
Consider CDMA with binary inputs and assume A for the spreading sequence.
Let $C_g$ be the capacity for Gaussian spreading sequences (symmetric i.i.d with unit variance). Then
$$
lim_{K\to +\infty}(C_K-C_g)=0
$$ 
\end{theorem}

This theorem turns
out to be very useful in order to obtain the bound on capacity because it allows
us to make use of convenient integration by parts identities that have no clear
counterpart in the non-Gaussian case. The proof of the theorem is given in
section~\ref{sec:indspreadingproof}.

\subsection{Existence of the limit $K\to +\infty$}
The interpolation method can be used to show the existence of the limit
$K\to +\infty$ for $C_K$. 

\begin{theorem}\label{thm:limitexist}
Consider CDMA with binary inputs and assume A for the spreading sequences
with uniform input distribution. Then 
\begin{align}
\lim_{K\to \infty} C_K \qquad \text{exists}
\end{align}
\end{theorem} 
The proof of this theorem is given in section \ref{sec:limitexistproof}
for Gaussian spreading sequences. The general case then follows because of Theorem \ref{thm:indspreading}.

\subsection{Tight upper bound on the capacity}\label{sec:upperbound}

The main result of this paper is that Tanaka's
formula (\ref{eqn:hrs}) is an upper bound to the capacity for all values of $\beta$.
\vskip 0.25cm
\begin{theorem}\label{thm:upperbound}
Consider CDMA with binary inputs and assume A for the spreading sequence. We have
\begin{align}\label{eqn:upperbound}
\lim_{K\to\infty}C_K\leq \min_{m\in [0,1]} c_{RS}(m)
\end{align}
where
$c_{RS}(m)$ is given by (\ref{eqn:hrs}). 
\end{theorem}
\vskip 0.25cm
If we combine this result with an inequality in Montanari and Tse
\cite{MoT06itw}, and exchanging
as they do the limits of $K\to +\infty$ and $sparse\to dense$, one can deduce
that the equality holds for some regime of noise
smaller than a critical value. This value corresponds to the threshold for
belief propagation decoding. Note that this equality is valid even if $\beta$ is
such that there is a phase transition (the fixed point equation
\eqref{eqn:fixedpt} has many solutions), whereas in \cite{MoT06itw} the equality holds for values of $\beta$ for which the phase transition does not occur.

Since the proof is rather complicated we find it useful to give the main ideas in an informal way.
The integral term in (\ref{eqn:hrs}) suggests that we can replace the original system with  a
simpler system where the user bits are sent through $K$ independent Gaussian
channels given by
\begin{align}
\tilde{y}_k = x_{k} + \frac{1}{\sqrt \lambda}w_k
\end{align}
where $w_k \sim \mathcal{N}(0,1)$ and 
$\lambda$ 
is an effective SNR.
Of course this argument is a bit naive because this effective system does not account for the extra terms in (\ref{eqn:hrs}), but it has the merit of identifying the correct interpolation.

We introduce an interpolating parameter $t\in [0,1]$ such that the independent Gaussian 
channels correspond to $t=0$ and the original CDMA system corresponds to $t=1$ (see Figure \ref{interpolationcomsystem})
\begin{figure}\label{interpolationcomsystem}
\begin{center}
\setlength{\unitlength}{0.42bp}
\begin{picture}(450,550)
\put(0,0){\includegraphics[scale=0.42]{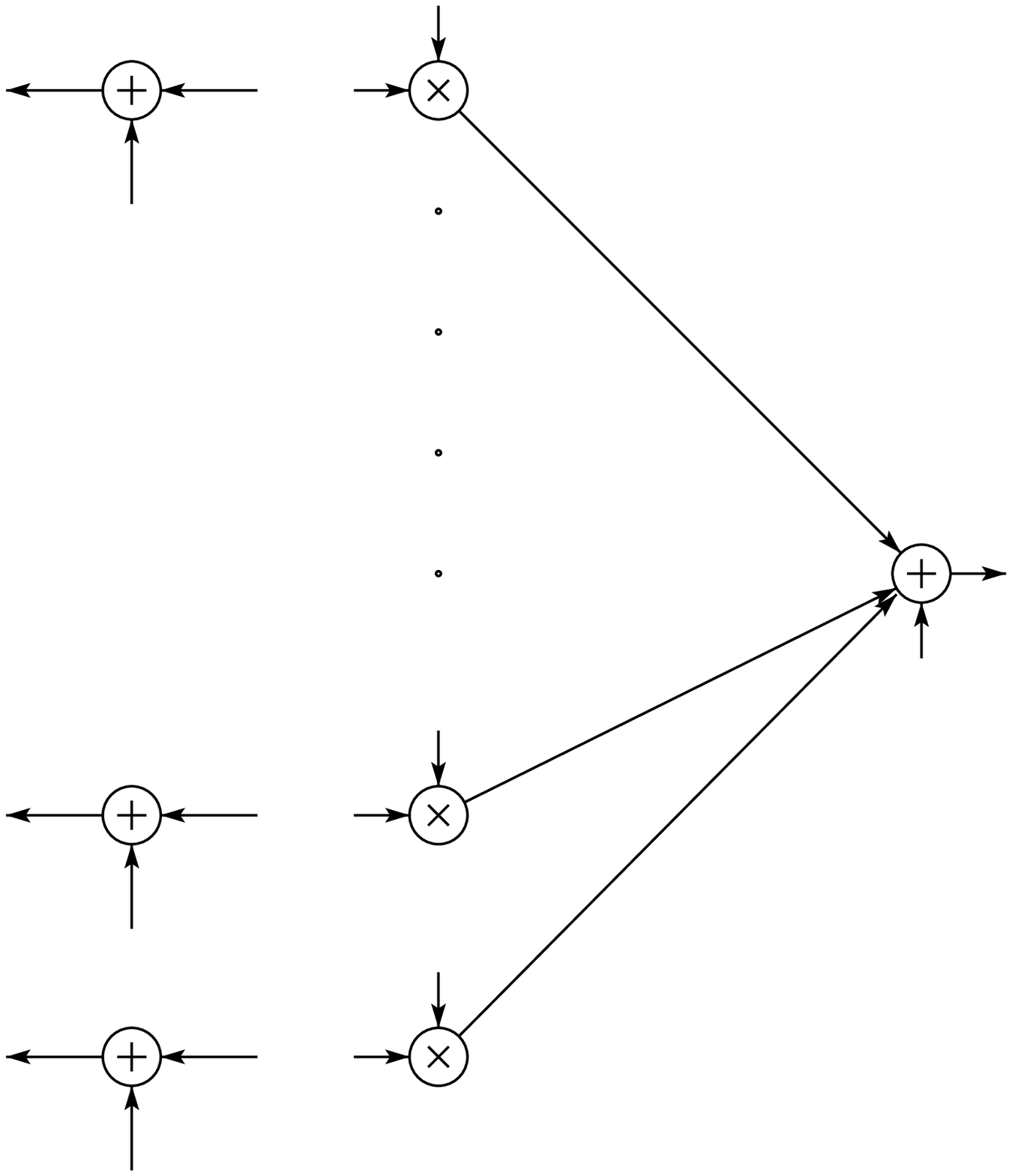}}
\put(200,555){\makebox(0,0)[t]{\small{$\underline{s}_{K}$}}}
\put(150,505){\makebox(0,0)[t]{\small{$x_{K}$}}}
\put(200,255){\makebox(0,0)[t]{\small{$\underline{s}_{2}$}}}
\put(150,205){\makebox(0,0)[t]{\small{$x_{2}$}}}
\put(200,155){\makebox(0,0)[t]{\small{$\underline{s}_{1}$}}}
\put(150,105){\makebox(0,0)[t]{\small{$x_{1}$}}}
\put(420,260){\makebox(0,0)[t]{\small{$\mathcal{N}(0,\frac{1}{B(t)})$}}}
\put(460,310){\makebox(0,0)[t]{\small{$\mby$}}}
\put(72,45){\makebox(0,0)[t]{\small{$\mathcal{N}(0,\frac{1}{\lambda(t)})$}}}
\put(72,145){\makebox(0,0)[t]{\small{$\mathcal{N}(0,\frac{1}{\lambda(t)})$}}}
\put(72,445){\makebox(0,0)[t]{\small{$ \mathcal{N}(0,\frac{1}{\lambda(t)})$}}}
\put(10,110){\makebox(0,0)[t]{\small{$\tilde y_{1}$}}}
\put(10,210){\makebox(0,0)[t]{\small{$\tilde y_{2}$}}}
\put(10,510){\makebox(0,0)[t]{\small{$\tilde y_{K}$}}}
\end{picture}
\caption{The information bits $x_k$ are transmitted through the normal
CDMA channel with variance $\frac1{B(t)}$ and through individual Gaussian channels with
noise $\frac1{\lambda(t)}$}
\end{center}
\end{figure}
It is convenient to denote the SNR of the original Gaussian channel as $B$ (that is $B=\sigma^{-2}$). Then \eqref{lambda} becomes 
$$
\lambda=\frac{B}{1+\beta B(1-m)}
$$
We introduce two interpolating SNR functions $\lambda(t)$ and $B(t)$ such that
\begin{equation}\label{eqn:lambdaBrelation}
\lambda(0)=\lambda,\,\, B(0)=0\qquad{\rm and}\qquad \lambda(1)=0,\,\, B(1)=B
\end{equation}
and 
\begin{align}\label{eqn:SNRrelation}
\frac{B(t)}{1+\beta B(t)(1-m)} + \lambda(t) = \frac{B}{1+\beta B(1-m)}
\end{align}
The meaning of \eqref{eqn:SNRrelation} is the following.
In the interpolating $t$-system the effective SNR seen by each user has an effective $t$-CDMA part and an independent channel part $\lambda(t)$ chosen such that the total SNR is fixed to the effective SNR of the CDMA system.
There is a whole class of interpolating functions satisfying the above conditions but it turns out that we do not need to specify them more precisely except for the fact that $B(t)$ is increasing, $\lambda(t)$ is decreasing and with continuous first derivatives. Subsequent calculations are independent of the particular choices of functions.

The parameter $m$ is to be considered as fixed to any arbitrary value in $[0,1]$. All the subsequent calculations are independent of its value, which is to be optimized to tighten the final bound.

We now have two sets of channel outputs $\mby$ (from the CDMA with noise variance $B(t)^{-1}$) and $\tilde\mby$ (from the independent channels with noise variance $\lambda(t)^{-1}$) and the interpolating communication system has a posterior distribution
\begin{equation}\label{eqn:postdistt}
p_t(\mbx\vert\mby, \tilde\mby, \mts)=\frac{1}{2^KZ(\mby,\tilde\mby,\mts)}
\exp\biggl(-\frac{B(t)}{2}\Vert \mby-N^{-{\frac{1}{2}}}\mts\mbx\Vert^2 - \frac{\lambda(t)}{2}\Vert \tilde\mby-\mbx\Vert^2\biggr) 
\end{equation}
Note that here we take without loss of generality $p_{\mbX}(\mbX)=\frac{1}{2^K}$.
By analyzing the mutual information $\bE_\mtS[I_t(\mbX ; \mbY,\tilde\mbY)]$ of the interpolating system we can relate 
$\bE_\mtS[I(\mbX;\mbY)]$ (the $t=1$ value) to the easily computed entropy $\bE_\mtS[I_0(\mbX ;\tilde\mbY)]$ of the independent channel limit. 
The average over $(\mbY,\tilde\mbY)$ is now performed with respect to
\begin{equation}\label{eqn:distyyt}
p_t(\mby,\tilde\mby\mid \mts)=\frac{1}{2^K}\sum_{\mbx^0}
\frac{1}{(\sqrt{2\pi B(t)^{-1}})^{N}(\sqrt{2\pi\lambda(t)^{-1}})^{K}}e^{-\frac{B(t)}{2}\Vert \mby-N^{-\frac{1}{2}}\mts\mbx^0\Vert^2-
 \frac{\lambda(t)}{2}\Vert \tilde\mby-\mbx^0\Vert^2} 
\end{equation}
These equations completely define the interpolating communication system.

In order to carry out this program successfully it turns out that we need
a concentration result on 
 empirical average of the ``magnetization'',
$$
m_1=\frac{1}{K}\sum_{k=1}^{K}x_k^0x_k
$$
which, as explained in section \ref{sec:concentration}, is closely related to the BER.
Informally speaking we need to prove that the fluctuations of 
$\bE\langle\vert m_1 - \bE\langle m_1\rangle\vert\rangle$
are small. This involves the control of two types of fluctuations,
$
\bE\langle\vert m_1 - \langle m_1\rangle\vert\rangle$ and 
$\bE\vert\langle m_1\rangle - \bE\langle m_1\rangle\vert$ (by the triangle inequality).
In some spin glass problems both type of fluctuations need not be small at the same time. Indeed it is a quite general fact that the first one is small for thermodynamic (or convexity) reasons while the smallness of the second is not assured if replica symmetry breaking occurs (see \cite{Tal03}). Here we use a crucial ingredient that is specific to the communication set up, namely the channel symmetry, which induces a gauge symmetry and prevents replica breaking. This, it turns out, allows to prove that both fluctuations are small.
The control of these fluctuations is the object of Theorem~\ref{thm:concmag} in
section~\ref{sec:conclemmas}. There are technical complications that we have to
deal with because such control of fluctuations is only possible away from phase
transitions. For this reason we have to add small appropriate perturbations to
the measure \eqref{eqn:postdistt} and give almost sure statements with respect to the strength of the perturbation. By being sufficiently careful with the order of limits the extra perturbation terms can be removed at the end of the calculations.

\section{Proof of bound on capacity: Theorem \ref{thm:upperbound}}\label{sec:upperboundproof}

\subsection{Preliminaries}
\label{sec:preliminaries}
The interpolating communication system defined by the measure \eqref{eqn:postdistt} allows us
to compare the original CDMA system with the independent channel system. The
distribution of $\mby,\tilde{\mby}$ is given by \eqref{eqn:distyyt}. This
distribution consists of a summation of $2^K$ terms, each corresponding to
different possible input sequence. Each of these terms contribute equally to the capacity (free
energy). The reader can explicitly check this by making the
change of variables $x_k\to x_k^0x_k$ and $s_{ik}\to s_{ik}x_k^0$, $w_k\to
w_kx_k^0$, $h_k\to h_kx_k^0$ which leave all standard Gaussians invariant. Hence we can assume that a
particular input sequence say $\mbx^0$  is transmitted. The distribution of the
received vectors with this assumption is 
\begin{equation}\label{eqn:distyytallone}
p_t(\mby,\tilde\mby\mid \mts)=
\frac{1}{(\sqrt{2\pi B(t)^{-1}})^{N}(\sqrt{2\pi\lambda(t)^{-1}})^{K}}e^{-\frac{B(t)}{2}\Vert
\mby-N^{-\frac{1}{2}}\mts\mbx^0\Vert^2-
 \frac{\lambda(t)}{2}\Vert \tilde\mby-\mbx^0\Vert^2}
\end{equation}
For
technical reasons that will become clear only in the next section we consider a
slightly more general interpolation system where the perturbation term 
\begin{align}\label{eqn:perturb}
h_u(\mbx)=\sqrt u \sum_{k=1}^K h_k x_k +u\sum_{k=1}^{K}x^0_kx_k -\sqrt u\sum_{k=1}^K|h_k|
\end{align}
is added in the exponent of the measure \eqref{eqn:postdistt}.
Here $h_k$ are i.i.d. $h_k \sim \mathcal{N}(0,1)$. For the moment $u\geq 0$ is arbitrary but in the sequel we will take $u\to 0$. 
This time it is convenient to perform a new change of variables 
$\mby=B(t)^{-1/2}\mbn+N^{-1/2}\mts\mbx^0$ and
$\tilde\mby=\lambda(t)^{-1/2}\mbw+\mbx^0$, where $n_i,w_i \sim \mathcal{N}(0,1)$ and we set $\langle -\rangle_{t,u}$ for the average  corresponding to the posterior measure
\begin{align}\label{eqn:perturbposterior}
p_{t,u}(\mbx\vert\mbn, \mbw,\mbh, \mts)=\frac{1}{Z_{t,u}}
\exp&\biggl(-\frac{1}{2}\Vert \mbn+N^{-{\frac{1}{2}}}B(t)^{\frac{1}{2}}\mts(\mbx^0-\mbx)\Vert^2
\\ \nonumber & 
- \frac{1}{2}\Vert \mbw+\lambda(t)^{\frac{1}{2}}(\mbx^0-\mbx)\Vert^2 + h_u(\mbx)\biggr) 
\end{align}
with the obvious normalization factor $Z_{t,u}$. We define a free energy
\begin{equation}
f_{t,u}(\mbn, \mbw,\mbh, \mts)=\frac{1}{K}\ln Z_{t,u}
\end{equation}
For $t=1$ we recover the original free energy,
$$
\bE[f(\mby,\mts)]=\frac{1}{2}+\lim_{u\to 0}\bE[f_{1,u}(\mbn,\mbw,\mbh,\mts)]
$$
while for $t=0$ the statistical sums decouple and we have the explicit result\footnote{it is also straightforward to compute the full $u$ dependence and see that it is $O(\sqrt u)$, uniformly in $K$}
\begin{align}\label{eqn:free0}
\frac{1}{2}+\lim_{u\to 0}\bE[f_{0,u}(\mbn,\mbw,\mbh,\mts)]= -\frac{1}{2\beta}- 
\lambda 
  + \int Dz \ln (2\cosh(\sqrt{\lambda} z+ \lambda)) 
\end{align}
where $\bE$ denotes the appropriate collective expectation over random objects.
In view of formula \eqref{eqn:entrofreerel} in order to obtain the average capacity it is sufficient to compute
\begin{equation}\label{limit}
\lim_{K\to+\infty}\lim_{u\to 0}\bE[f_{1,u}(\mbn, \mbw,\mbh, \mts)] +\frac{1}{2}
\end{equation}
There is no loss in generality in setting 
\begin{equation}\label{inputone}
x_k^0=1
\end{equation}
 for the input symbols.  From
now on in sections \ref{sec:upperboundproof},\ref{sec:concmag}, and
\ref{sec:limitexistproof} we stick to \eqref{inputone}.
We also use the shorthand notations
$$
z_k=x_k^0-x_k=1-x_k,\qquad f_{t,u}(\mbn, \mbw,\mbh, \mts)=f_{t,u}
$$ 
Using $\vert h_u(\mbx)\vert\leq 2\sqrt u\sum_k\vert h_k\vert  + K u$ it easily follows that ($u$ small)
\begin{align}\label{eqn:freecontu}
\vert \bE[f_{t,u}] - \bE[f_{t,0}] \vert \leq 2\sqrt u\bE[\vert h_k\vert] + u
\end{align}
therefore we can permute the two limits in \eqref{limit} and compute
$$
\lim_{u\to 0}\lim_{K\to+\infty}\bE[f_{1,u}] +\frac{1}{2}
$$
From now on we keep the limits in that order.
By the fundamental theorem of calculus,
\begin{equation}\label{fund}
\bE[f_{1,u}] = \bE[f_{0,u}] + \int_0^1dt \frac{d}{dt}\bE[f_{t,u}]
\end{equation}
Our task is now reduced to estimating
$$
\lim_{u\to 0}\lim_{K\to+\infty}\int_0^1dt \frac{d}{dt}\bE[f_{t,u}]
$$
This is done in sections \ref{sec:derivative}, \ref{sec:endofproof}. This requires a few preliminary results that are the object of sections \ref{sec:nish}, \ref{sec:conclemmas}.

\subsection{Nishimori identities}\label{sec:nish}

As already alluded to in the introduction the ``magnetization" plays an important role
 \begin{equation}
 m_1=\frac{1}{K}\sum_{k=1}^{K} x_k
 \end{equation}
A closely related quantity is the ``overlap parameter"
\begin{equation}
q_{12}=\frac{1}{K}\sum_{k=1}^K x_k^{(1)}x_k^{(2)}
\end{equation}
where $x_k^{(1)}$ and
$x_k^{(2)}$ are independent copies (``replicas") of the $x_k$. This means that
the joint distribution of $(x_k^{(1)}, x_k^{(2)})$ is the product measure 
$$
p_t(\mbx^{(1)}\vert \mbn,\mbw,\mbh,\mts)p_t(\mbx^{(2)}\vert \mbn,\mbw,\mbh, \mts)
$$
 The average
 with respect to this joint distribution is denoted (by
a slight abuse of notation) with the same bracket $\langle - \rangle_{t,u}$. The
important thing to notice is that the replicas are ``coupled" through the common
randomness $(\mbn,\mbw,\mbh,\mts)$.

\begin{lemma}\label{lem:nishimorimq}
The distributions of $m_1$ and $q_{12}$  defined as
$$
\bP_{m_1}(x) = \bE \langle \delta(x-m_1) \rangle_{t,u},\quad \bP_{q_{12}}(x) =
\bE\langle\delta(x-q_{12})\rangle_{t,u}
$$
are equal, namely 
$$
\bP_{m_1}(x) = \bP_{q_{12}}(x)
$$
\end{lemma}
In particular the following identity holds
\begin{equation}\label{eqn:mq}
\bE[\langle m_1\rangle_{t,u}]=\bE[\langle q_{12}\rangle_{t,u}]
\end{equation} 
Such identities are known as Nishimori identities in the statistical physics
literature and are a  consequence of a gauge symmetry satisfied by the measure
$\bE\langle-\rangle_{t,u}$. They have also been used in the context of
communications (see \cite{Mon01},\cite{Mon05}). For completeness a sketch of the proof is
given in Appendix \ref{apen:nishimori}. 

The next two identities also follow from similar considerations. 
\begin{lemma}\label{lem:nishimori2}
Let
$$
\mcZ= \mbn + \sqrt{\frac{B(t)}{N}}\mts \mbz
$$
Consider two replicas $\mcZ^{(\alpha)}$, $\alpha=1,2$ corresponding to $z_k^{(\alpha)}=1-x_k^{(\alpha)}$. We have
then
\begin{equation}\label{eqn:X11}
\frac{1}{N}\bE[\langle \Vert\mcZ\Vert^2\rangle_{t,u}] = 1
\end{equation}
and 
\begin{align}\label{eqn:X12}
\bE[\langle
(\mbn\cdot\mcZ^{(2)})(\mbz^{(1)}\cdot\mbz^{(2)})\rangle_{t,u}] 
= \sum_{k}\bE[\langle
(\mbn\cdot\mcZ)z_k\rangle_{t,u} ]
\end{align}
\end{lemma}

\subsection{Concentration of Magnetization}\label{sec:conclemmas}

A crucial feature of the calculation in the next paragraph is that $m_1$ (and $q_{12}$) concentrate, namely
\vskip 0.25cm
\begin{theorem}\label{thm:concmag}
Fix any $\epsilon>0$. For Lebesgue almost every $u>\epsilon$, 
\begin{align*}
\lim_{N\to\infty}\int_0^1 dt\bE\langle\vert m_1 - \bE\langle m_1\rangle_{t,u}\vert\rangle_{t}  = 0
\end{align*}
\end{theorem}
\vskip 0.25cm
The proof of this theorem, which is the point where the careful tuning of the
perturbation is needed, has an interest of its own and is presented section
\ref{sec:concmag}. Similar statements in the spin glass literature have
been obtained by Talagrand \cite{Tal03}. The usual signature of replica symmetry breaking is the absence of concentration for the overlap parameter $q_{12}$. This theorem combined with the Nishimori identity 
``explains" why the replica symmetry is not broken.

We will also need the following corollary
\begin{corollary}\label{cor:concmag}
The following holds
\begin{align*}
\frac{1}{N^{3/2}}\bE\langle(\mbn\cdot\mts\mbz)(1-m_1)\rangle_{t,u}=
\frac{1}{N^{3/2}}\bE\langle\mbn\cdot
\mts\mbz\rangle_{t,u} (1-\bE\langle m_1\rangle_{t,u}) + o_N(1)
 \end{align*}
with $\lim_{N\to +\infty}o_N(1) =0$ for almost every $u > 0$.
\end{corollary}
\begin{proof} 
By the Cauchy-Schwartz inequality
\begin{align*}
\frac{1}{N^{3/2}}\bE\langle(\mbn\cdot
\mts\mbz)(\bE\langle m_1 \rangle_{t,u}-m_1 )\rangle_{t,u}
\leq &
\frac{1}{N^{3/2}}
(\bE\langle(\mbn\cdot
\mts\mbz)^2\rangle_{t,u})^{1/2}\\
\nonumber & \times (\bE\langle (\bE\langle m_1 \rangle_{t,u} - m_1)^2 \rangle_{t,u})^{1/2} 
\end{align*}
Because of the concentration of the magnetization $m_1$ (theorem \ref{thm:concmag}) it suffices to prove that 
\begin{equation}\label{eqn:centrallimbound}
\bE\Big\langle \Big({N^{-\frac32}}\sum_{i,l} n_i s_{il}
z_l\Big)^2\Big\rangle_{t,u} \leq D
\end{equation}
for some constant $D$ independent of $N$.
The proof follows from the central limit theorem and is given in Appendix \ref{apen:centrallimbound}.
\end{proof}

\subsection{Computation of $\frac{d}{dt}\bE[f_{t,u}]$}
\label{sec:derivative}
We have 
\begin{equation}\label{derivativesplit}
\frac{d}{dt}\bE[f_{t,u}] = T_1 + T_2
\end{equation}
where 
\begin{align}
T_1=-\frac{\lambda^\prime(t)}{2\sqrt{\lambda(t)}K}\bE\langle
\mbw\cdot\mbz\rangle_{t,u} -  \frac{\lambda^\prime(t)}{2K}\bE\langle
\mbz\cdot\mbz\rangle_{t,u} 
\end{align}
and
\begin{align}
T_2 = -
\frac{1}{K\sqrt N}\frac{B^\prime(t)}{2\sqrt{B(t)}}
\bE\langle\mcZ\cdot
\mts\mbz\rangle_{t,u} 
\end{align}

\subsubsection{Transforming $T_1$}

Integration by parts with respect to $w_k$ leads to
\begin{align*}
&T_1  =
\frac{\lambda^\prime(t)}{2\sqrt{\lambda(t)}K}\bE\langle(\mbw
+\sqrt{\lambda(t)}\mbz)\cdot \mbz\rangle_{t,u}
\\
&-\frac{\lambda^\prime(t)}{2\sqrt{\lambda(t)}K}\bE\langle \mbz^{(1)}\cdot(\mbw+\sqrt{\lambda(t)}\mbz^{(2)})\rangle_{t,u}
- \frac{\lambda^\prime(t)}{2K}\bE\langle \mbz\cdot\mbz\rangle_{t,u} \\
&=- \frac{\lambda^\prime(t)}{2}\bE\langle 1-2m_1+q_{12}\rangle_t = - \frac{\lambda'(t)}{2}\bE\langle1-m_1\rangle_{t,u}
\end{align*}
To obtain the second equality we remark that the $\mbw$ terms cancel and for the third one follows from \eqref{eqn:mq}.
 From the relation between $\lambda(t)$ and $B(t)$ given in
equation (\ref{eqn:SNRrelation}), $T_1$ can be rewritten in the form 
\begin{align}\label{eqn:T1}
T_1 = \frac{B'(t)}{2(1+\beta (1-m)B(t))^2}\bE \langle 1- m_1\rangle_{t,u}
\end{align}

\subsubsection{Transforming $T_2$}
\label{sec:transformingT2}
The term
$T_2$ can be rewritten as
\begin{align*}
T_2&=-\frac{B^\prime(t)}{2\beta B(t)N}\bE\langle \Vert\mcZ\Vert^2\rangle_{t,u} + {\frac{B^\prime(t)}{2\beta
B(t)N}} \bE\Vert \mbn\Vert^2 \\ &+ \frac{B^\prime(t)}{2
\sqrt{B(t)}K\sqrt N}\bE\langle \mbn\cdot \mts\mbz\rangle_{t,u}
\end{align*}
Because of \eqref{eqn:X11} the first two terms cancel,
\begin{align}\label{simpleform}
T_2= \frac{B^\prime(t)}{2
\sqrt{B(t)}K\sqrt N}\bE\langle\mbn\cdot \mts\mbz \rangle_{t,u}
\end{align}
Now we use integration by parts with respect to
$s_{ik}$, 
\begin{align*}
T_2 = -\frac{B^\prime(t)}{2KN}\bE\langle
(\mbn\cdot\mcZ)(\mbz\cdot \mbz)\rangle_{t,u}+
\frac{B^\prime(t)}{2 K N}\bE\langle
(\mbn\cdot\mcZ^{(2)})(\mbz^{(1)}\cdot \mbz^{(2)})\rangle_{t,u} 
\end{align*}
and the Nishimori identity \eqref{eqn:X12}
\begin{align*}
T_2 = &-\frac{B^\prime(t)}{2KN}\sum_{k}\bE\langle
(\mbn\cdot\mcZ)z_k\rangle_{t,u}
\\&=
-\frac{ B^\prime(t)}{2}\frac{1}{NK}\sum_k\bE[\Vert\mbn\Vert^2\langle z_k\rangle_{t,u}] \\
& -
\frac{ B^\prime(t)\sqrt{B(t)}}{2K N^{3/2}}\sum_k\bE\langle
(\mbn\cdot \mts\mbz)(1-x_k)\rangle_{t,u}
\end{align*}
Since $\frac{1}{N}\Vert\mbn\Vert^2=\frac{1}{N}\sum_i n_i^2$ concentrates on $1$, we get
\begin{align*}
T_2 = & -\frac{B^\prime(t)}{2}\beta \bE\langle 1-m_1\rangle_{t,u} + o_N(1)\\&
- \frac{\beta B^\prime(t)\sqrt{B(t)}}{2K N^{1/2}}\bE\langle
(\mbn\cdot \mts\mbz)(1-m_1)\rangle_{t,u}
\end{align*}
Applying Corollary \ref{cor:concmag} to the last expression for $T_2$ together with \eqref{simpleform} we obtain a closed affine equation for the later, whose solution is
\begin{align}\label{eqn:T2}
T_2= -\frac{B^\prime(t) \bE\langle1-m_1 \rangle_t}{2(1+\beta B(t)\bE\langle
1-m_1\rangle_{t,u})} + o_N(1)
\end{align}

\subsection{End of proof}
\label{sec:endofproof}
We
add and subtract the term  
$\frac{1}{2\beta}\ln(1+\beta B(1-m))$ from \eqref{fund} and use the integral representation
\begin{align*}
\frac{1}{2\beta}\ln(1+\beta B(1-m))=\frac{1}{2\beta}\int_0^1 dt \frac{\beta B'(t)(1-m)}{1+\beta B(t) (1-m)}
\end{align*} 
to obtain
\begin{equation}\nonumber
 \bE[f_{1,u}] =\bE[f_{0,u}]-\frac{1}{2\beta}\ln(1+\beta B(1-m)) +
\int_0^1 dt\biggl(\frac{d}{dt}\bE[f_{t,u}]+\frac{ B'(t)(1-m)}{2(1+\beta B(t) (1-m))}\biggr)
\end{equation}
If one uses \eqref{derivativesplit} and expressions (\ref{eqn:T1}), (\ref{eqn:T2}) some remarkable algebra occurs in the last integral. The integrand becomes
$$
R(t)+ \frac{B^\prime(t)(1-m)}{2(1+\beta B(t)(1-m))^2}
$$
with
\begin{align*}
R(t)=\frac{\beta B'(t)B(t)(\bE\langle m_1 - m \rangle_{t,u})^2}{2 (1+\beta B(t)
(1-m))^2(1+\beta B(t) \bE\langle 1-m_1\rangle_{t,u})}
\end{align*}
So the integral has a positive contribution $\int_0^1 dt R(t)\geq 0$ plus a computable contribution equal to $\frac{B(1-m)}{2(1+\beta B(1-m))}=\frac{\lambda}{2}(1-m)$.
Finally thanks to \eqref{eqn:free0} we find
\begin{align}
\frac{1}{2}+\bE[f_{1,u}] & =\int Dz \ln (2\cosh(\sqrt{\lambda} z+ \lambda))  - \frac{1}{2\beta} - \frac{1}{2\beta}\ln (1+\beta
B(1-m)) \nonumber \\
&
  -\frac{\lambda}{2}(1+m) + \int_0^1 R(t) dt +o_N(1)+ O(\sqrt u) \label{eqn:free1}
\end{align}
where for a.e $u>\epsilon$, $\lim_{N\to \infty} o_N(1)=0$.
We take first the limit $N\to \infty$, then $u\to\epsilon$ (along some appropriate sequence) and then $\epsilon\to 0$ to obtain a formula for the free energy where the only non-explicit contribution is $\int_0^1dt R(t)$. Since this is positive for all $m$, we obtain a lower bound on the free energy which is equivalent to the announced upper bound on the capacity. 

\section{Concentration of Magnetization}\label{sec:concmag}

The goal of this section is to prove Theorem \ref{thm:concmag}. 
The proof is organized in a succession of lemmas.
By the same methods used for Theorem \ref{thm:freeconc} we can prove 
\vskip 0.25cm
\begin{lemma}\label{lem:freetuconc}
There exists a strictly positive constant $\alpha$ (which remains positive for all $t$ and $u$) such that
\begin{equation*}
\bP[\vert f_{t,u}- \bE[f_{t,u}]\vert\geq \epsilon ] =O(e^{-\alpha \epsilon^2 \sqrt K})
\end{equation*}
\end{lemma}
\vskip 0.25cm

The perturbation term (\ref{eqn:perturb}) has been chosen carefully so that the following holds,
\vskip 0.25cm
\begin{lemma}\label{lem:ftuconvex}
When considered as a function of $u$, $f_{t,u}$ is convex in $u$.
\end{lemma}
\vskip 0.25cm
\begin{proof}
We simply evaluate the second derivative and show it is positive.
\begin{align*}
\frac{df_{t,u}}{du} = \langle L(\mbx)\rangle_{t,u}-\frac{1}{K 2\sqrt u }\sum_k |h_k|
\end{align*}
where we have defined
\begin{equation*}
L(\mbx) = \frac1K\frac{1}{2\sqrt u} \sum_k h_kx_k + \frac1K\sum_kx_k
\end{equation*}
Differentiating again, 
\begin{align}\label{2derivative}
\frac{d^2f_{t,u}}{du^2} &= \frac{1}{K}\Big \langle\frac{-1}{4 u^{3/2}}\sum_k h_kx_k \Big\rangle_{t,u}+\frac{1}{4u^{3/2}K} \sum_k |h_k|\nonumber\\& +K(\langle L(\mbx)^2\rangle_{t,u}-\langle L(\mbx)\rangle^2_{t,u})\geq 0
\end{align}
\end{proof}
\vskip 0.25cm
The quantity $L(\mbx)$ turns out to be very useful and satisfies two concentration properties.
\vskip 0.25cm
\begin{lemma} For any $a>\epsilon>0$ fixed,
\begin{align*}
\int_{\epsilon}^{a} du \bE\Big\langle &\Big|L(\mbx)
-\langle L(\mbx)\rangle_{t,u}\Big|\Big\rangle_{t,u}=
O\Big(\frac1{\sqrt K} \Big)
\end{align*}
\end{lemma}
\vskip 0.25cm
\begin{proof}
From equation (\ref{2derivative}), we have 
\begin{align*}
\int_\epsilon^a du \bE\Big\langle &\Big(L(\mbx)-\langle
L(\mbx)\rangle_{t,u}\Big)^2\Big\rangle_{t,u} \leq \int_\epsilon^a du
\frac1K\frac{d^2}{du^2}\bE[f_{t,u}] \\
&\leq \frac1K \Big(\frac{d}{du}\bE[f_{t,a}]  - \frac{d}{du}\bE[f_{t,\epsilon}]\vert\Big)
= O\Big(\frac1K \Big)
\end{align*}
In the very last equality we use that the first derivative of $\bE[f_{t,u}]$ is bounded for $u\geq \epsilon$.
Using Cauchy-Schwartz inequality for $\int \bE\langle- \rangle_{t,u}$ we obtain the lemma.
\end{proof}
\vskip 0.25cm
\begin{lemma} For any $a > \epsilon >0$ fixed,
\begin{align*}
\int_{\epsilon}^a du \bE\Big|\langle L(\mbx)\rangle_{t,u} - \bE\langle
L(\mbx)\rangle_{t,u}\Big| = O\Big(\frac1{K^{\frac1{16}}}\Big)
\end{align*}
\end{lemma}
\vskip 0.25cm
\begin{proof}
From convexity of $f_{t,u}$ with respect to $u$
(lemma \ref{lem:ftuconvex}) we have for any $\delta>0$,
\begin{align*}
\frac{d}{du}f_{t,u} &- \frac{d}{du}\bE[f_{t,u}]
\leq \frac{f_{t,u+\delta} - f_{t,u}}{\delta} - \frac{d}{du} \bE[f_{t,u}] \\
& \leq \frac{f_{t,u+\delta} - \bE[f_{t,u+\delta}]}{\delta} - \frac{f_{t,u}-\bE[f_{t,u}]}{\delta} 
\\&\;\;\;\;\;+ \frac{d}{du}\bE[f_{t,u+\delta}] - \frac{d}{du}\bE[f_{t,u}]
\end{align*}
A similar lower bound holds with $\delta$ replaced by $-\delta$. 
Now from Lemma \ref{lem:freetuconc} we know that the first two 
terms are $O(K^\frac14)$.
Thus from the formula for the first derivative in the proof of Lemma
\ref{lem:ftuconvex} and the fact that the fluctuations of $\frac{1}{K}\sum_{k=1}^{K} \vert h_k\vert$ are $O(\frac{1}{\sqrt K})$ we get
\begin{align*}
\bE\Big|\langle L(\mbx)\rangle_{t,u} -
\bE\langle L(\mbx) \rangle_{t,u}\Big|
\leq  \frac{1}{\delta}O\Big(\frac{1}{\sqrt K}\Big)+\frac{1}{\delta}O\Big(\frac{1}{K^{\frac14}}\Big)
\\+ \frac{d}{du}\bE[f_{t,u+\delta}] - \frac{d}{du}\bE[f_{t,u}]
\end{align*}
We will choose $\delta = \frac{1}{K^{\frac18}}$.
Note that we cannot assume that the difference of the two derivatives is small because the first derivative of the free energy is not uniformly continuous in $K$ (as $K\to\infty$ it may develop jumps at the phase transition points).  The free  energy itself is uniformly continuous.
For this reason if we integrate with respect to u, using \eqref{eqn:freecontu}
we get
\begin{align*}
\int_{\epsilon}^adu\bE\Big|\langle L(\mbx)\rangle_{t,u} - \bE\langle
L(\mbx)\rangle_{t,u}\Big| \leq O\Big(\frac1{K^\frac{1}{16}}\Big)
\end{align*}
\end{proof}

Using the two last lemmas we can prove Theorem \ref{thm:concmag}.
\vskip 0.25cm
\noindent {\it Proof of Theorem \ref{thm:concmag}: }
Combining the concentration lemmas we get
\begin{align*}
\int_{\epsilon}^a du \bE\langle| L(\mbx) - \bE\langle
L(\mbx)\rangle_{t,u}|\rangle_{t,u}
 \leq O\Big(\frac{1}{K^{\frac{1}{16}}}\Big)
\end{align*}
For any function $g(\mbx)$ such that $|g(\mbx)| \leq 1$, we have 
\begin{align*}
\int_{\epsilon}^a du &|\bE\langle L(\mbx)g(\mbx)\rangle_{t,u} - \bE\langle
L(\mbx)\rangle_{t,u} \bE\langle g(\mbx)\rangle_{t,u}|\rangle_{t,u}
\leq \int_{\epsilon}^a du\bE\langle |L(\mbx) - \bE\langle
L(\mbx)\rangle_{t,u}|\rangle_{t,u}
\end{align*}
More generally the same thing holds if one takes a function depending on many replicas such as $g(\mbx^{(1)}, \mbx^{(2)})= q_{12}$.
Using integration by parts formula with respect to $h_k$,
\begin{align}\label{eqn:lhslq}
&\bE\langle L(\mbx)q_{12}\rangle_{t,u} = \bE\Big\langle\frac{1}{2K\sqrt u}
\sum_k h_k x_k q_{12}\Big\rangle_{t,u} + \bE\langle m_1 q_{12}
\rangle_{t,u}\nonumber \\ &
= \frac{1}{2}\bE\langle (1+ q_{12})q_{12}\rangle_{t,u} - \frac{1}{2}\bE\langle
(q_{13}+ q_{14})q_{12}\rangle_{t,u} + \bE\langle m_1 q_{12}\rangle_{t,u} \nonumber \\
& = \frac{1}{2}\bE\langle (1+ q_{12})q_{12}\rangle_{t,u} = \frac12\bE\langle
m_1 + m_1^2\rangle_{t,u} 
\end{align}
where in the last two equalities we used the Nishimori identity (\ref{eqn:mq}). By a similar calculation,
\begin{align}\label{eqn:rhslq}
\bE\langle L(\mbx) \rangle_{t,u}\bE\langle q_{12} \rangle_{t,u} & =
\frac{1}{2}\bE\langle 1 - q_{12} + 2m_1\rangle_{t,u} \bE\langle
q_{12}\rangle_{t,u} \nonumber \\
& = \frac12(\bE\langle m_1\rangle_{t}+(\bE\langle m_1 \rangle_{t})^2)
\end{align}
From equations (\ref{eqn:lhslq}) and (\ref{eqn:rhslq}), we get
\begin{align*}
\int_{\epsilon}^a du |\bE\langle m_1^2\rangle_{t,u} - (\bE\langle
m_1\rangle_{t,u})^2| \leq O\Big(\frac{1}{K^{\frac{1}{16}}}\Big)
\end{align*}
Now integrating with respect to $t$ and exchanging the integrals (by
Fubini's theorem), we get
\begin{align*}
\int_{\epsilon}^a du \int_0^1 dt |\bE\langle m_1^2\rangle_{t,u} - (\bE\langle
m_1\rangle_{t,u})^2| \leq O\Big(\frac{1}{K^\frac{1}{16}}\Big)
\end{align*}
The limit of the left hand side as $K\to\infty$ therefore vanishes. By Lebesgue's theorem this limit can be exchanged with the $u$ integral and we get the desired result. (Note that one can further exchange the limit with the $t$-integral and 
obtain that the fluctuations of $m_1$ vanish for
almost every $(t,u)$).

\section{Proof of independence from spreading sequence distribution: Theorem
\ref{thm:indspreading}}
\label{sec:indspreadingproof}

We consider a communication system with spreading values $r_{ik}$ generated from
a symmetric distribution with unit variance and satisfying assumption A. We
compare the capacity of this system to the Gaussian ${\cal N}(0,1)$ case whose
spreading sequence values are denoted by $s_{ik}$. The comparison is done
through an interpolating system with respect to the two spreading sequences
$$
v_{ik}(t)=\sqrt t r_{ik} + \sqrt{1-t} s_{ik}, \qquad 0\leq t\leq 1
$$
Let $\mtv(t)$ denote the matrix with entries $v_{ik}(t)$ and let $\mbv_i(t)$ denote the $i$th row of the matrix.
By the fundamental theorem of calculus the capacities are related by
$$
C_K-C_g=\bE_{\mtR}[C(\mtr)]-\bE_{\mtS}[C(\mts)]=\int_0^1dt \frac{d}{dt}\bE_{\mtV(t)}[C(\mtv(t))]
$$
From \eqref{eqn:entrofreerel} the derivative is equal to
$$
\frac{d}{dt}\bE_{\mtV(t)}[C(\mtv(t))] = - \bE_{\mtS}\bE_{\mtR}\frac{d}{dt} \bE_{\mbY\vert \mtV(t)}[f(\mby, \mtv(t)]
$$
As before we can assume that the transmitted sequence is $\mbs^0$. It is
convenient to first perform the change of variables 
$\mby = \mbn + N^{-1/2}\mtv(t) \mbx^{0}$ and then perform the $t$ derivative. One finds
\begin{align}\label{eqn:derivative}
\frac{d}{dt}\bE_{\mtV(t)}[C(\mtv(t))]  = \frac{1}{\sigma^2 K\sqrt N} \bE_{\mtS,\mtR,\mbN}\Big\langle \biggl(\mbn + \frac{1}{\sqrt N}\mtv(t) (\mbx^0-\mbx)\biggr)\cdot
\mtv^\prime(t)(\mbx^0-\mbx)  \Big\rangle_t
\end{align}
where $\langle-\rangle_t$ is the average with respect to the normalized measure
$$
\frac{1}{2^KZ_t}\exp(-\frac{1}{2\sigma^2}\Vert \mbn - N^{-\frac{1}{2}}\mtv(t)(\mbx^0-\mbx)\Vert^2)
$$
We split \eqref{eqn:derivative} in two contributions $T_1-T_2$ corresponding to
\begin{equation}\label{twocont}
\mtv^\prime(t)=\frac{1}{2\sqrt t}\mtr-\frac{1}{2\sqrt {1-t}}\mts
\end{equation}
For $T_1$ we have
\begin{align}\label{sum1}
T_1
 = \sum_{i,k}T_1(i,k) =\frac{1}{2\sqrt{t}}\sum_{i,k} \bE_{\mtS,\mtR,\mbN} [r_{ik}g_{ik}]
\end{align}
with 
\begin{align}\label{eqn:gik}
g_{ik}= \frac{1}{\sigma^2 K\sqrt N} \Big\langle \biggl(\mbn + \frac{1}{\sqrt N}\mtv(t) (\mbx^0-\mbx)\biggr)_i
(x^0_k-x_k)  \Big\rangle_t
\end{align}
For $T_2$ we have 
\begin{align}\label{sum2}
T_2 = \sum_{i,k}T_2(i,k) =\frac{1}{2\sqrt {1-t}}\sum_{i,k} \bE_{\mtS,\mtR,\mbN} [s_{ik} g_{ik}]
\end{align}
with the same expression for $g_{ik}$.
For each contribution in the sums \eqref{sum1}, \eqref{sum2} we use integration by parts formulas. For \eqref{sum1} we use the formula (it is an exercise to check that it is valid for any symmetric random variable)
\begin{align}
\label{eqn:symintparts}
\bE[ r_{ik} g(r_{ik})] & 
= \bE[r_{ik}^2 \frac{\partial g(r_{ik})}{\partial r_{ik}}] -
\frac{1}{4}\bE\biggl[\vert r_{ik}\vert\int_{-\vert r_{ik}\vert}^{\vert
r_{ik}\vert}(r_{ik}^2-u^2) \frac{\partial^3 g(u)}{\partial u^3} du\biggr]
\nonumber
\\
= \bE[\frac{\partial g(r_{ik})}{\partial r_{ik}}] & +
\bE\biggl[(r_{ik}^2-1)\int_0^{r_{ik}} \frac{\partial ^2 g(u)}{\partial u^2} du\biggr] 
-\frac{1}{4}\bE\biggl[\vert r_{ik}\vert\int_{-\vert r_{ik}\vert}^{\vert
r_{ik}\vert}(r_{ik}^2-u^2) \frac{\partial^3 g(u)}{\partial u^3} du\biggr]
\end{align}
and for \eqref{sum2} we use the standard Gaussian (unit variance) integration 
by parts formula 
\begin{align}\label{eqn:gaussintparts}
\bE[ s_{ik} g(s_{ik})] &= \bE[ \frac{\partial g(s_{ik})}{\partial s_{ik}}]
\end{align}
When we consider $T_1-T_2$ the term corresponding to the expectation in
\eqref{eqn:gaussintparts} cancels with that of the  first expectation in \eqref{eqn:symintparts} and we get
\begin{align}\label{eqn:T1-T2}
T_1-T_2  =&\frac{1}{2\sqrt t}\sum_{i,k} \bE\biggl[(r_{ik}^2-1)\int_0^{r_{ik}} \frac{\partial ^2
g_{ik}(u)}{\partial u^2} du\biggr] 
-\frac{1}{8\sqrt{t}}\sum_{i,k}\bE\biggl[\vert r_{ik}\vert\int_{-\vert r_{ik}\vert}^{\vert
r_{ik}\vert}(r_{ik}^2-u^2) \frac{\partial^3 g_{ik}(u)}{\partial u^3} du\biggr]
\end{align}
It remains to prove that both terms with the partial derivatives tend to zero as
$N\to +\infty$. This computation is rather lengthy and is deferred to Appendix
\ref{appen:independence}, but for the convenience of the reader we point out the mechanism that
is at work. On the expression for $g_{ik}$ one sees that when the
$\frac{\partial^2}{\partial u^2}$ and $\frac{\partial^3}{\partial u^3}$
derivatives are performed extra powers $N^{-1}$ and $N^{-3/2}$ are generated. 
Therefore we get
\begin{equation}\label{eqn:indsecondderivative}
\bE\biggl[(r_{ik}^2-1)\int_0^{r_{ik}} \frac{\partial ^2 g_{ik}}{\partial u_{ik}^2} du_{ik}\biggr]=O(N^{-5/2})
\end{equation}
and 
\begin{equation}\label{eqn:indthirdderivative}
\bE\biggl[\vert r_{ik}\vert\int_{-\vert r_{ik}\vert}^{\vert
r_{ik}\vert}(r_{ik}^2-u_{ik}^2) \frac{\partial^3 g}{\partial u_{ik}^3} du_{ik}\biggr] = O(N^{-3})
\end{equation}
Since one sums over $KN$ terms one finds that the final contributions are $O(N^{-1/2})$ and $O(N^{-1})$.

\section{Proof of existence of limit : Theorem \ref{thm:limitexist}}
\label{sec:limitexistproof}
Let us recall the following relation between the free energy and the capacity.
\begin{align}
C_K = \frac{1}{2\beta} - \bE[f(\mby,\mts)]
\end{align}
where $f(\mby,\mts)$ is defined in \eqref{eqn:freeenergy} with $p_{\mbX}(\mbx) = \frac{1}{2^K}$. 
This implies that  it is sufficient to show the existence of limit for the average free energy $\free_K = \bE[f(\mby,\mts)]$.
The theorem is proved by showing that the sequence $K \free_K$ is super additive, $K \free_K \geq K_1 \free_{K_1} + K_2 \free_{K_2}$ for $K = K_1 + K_2$. From standard theorems it then follows that the 
limit $\free_K$ exists. As in the previous sections, working directly with this system is difficult and hence
we perturb the Hamiltonian with $h_u(\mbx)$ as defined in \eqref{eqn:perturb}.
\begin{align}
H_u(\mbx) = & -\frac{1}{2\sigma^2} \Vert \mbn + \frac{1}{\sqrt{N}}\mts
(\mbone-\mbx) \Vert^2 + h_u(\mbx)
\end{align}
Let us define the corresponding partition function as $Z_u$ and the free energy as $\free_K(u) = \frac{1}{K}\bE[\ln Z_u]$.
The original free energy is obtained by substituting $u=0$, i.e., $\free_K = \free_K(0)$. From the uniform continuity of $\free_K(u)$, it is sufficient to show the convergence of $\free_K(u)$ for some $u$ close to zero. Even this turns out to be difficult and
what we can show is the existence of the limit $\int_{u=\epsilon}^{a} \free_K(u)
du$ for any $a > \epsilon > 0$. However this is sufficient for us due to the
following: from the continuity of the free energy with $u$ \eqref{eqn:freecontu}
we have 
\begin{align*}
\int_{\epsilon}^{2\epsilon} (\free_K(u) - |O(1)|\sqrt{u}) du \leq \epsilon \free_K \leq \int_{\epsilon}^{2\epsilon} (\free_K(u) + |O(1)|\sqrt{u})  du  
\end{align*}
Since the limit of the integral exists, we have
\begin{align*}
|\limsup_{K\to\infty}\free_K - \liminf_{K\to\infty}\free_K| \leq |O(1)|\sqrt{\epsilon}
\end{align*}
This $\epsilon$ can be made as small as desired and hence the theorem follows.

Let $K = K_1 + K_2$ and let $\frac{K}{\beta}, \frac{K_1}{\beta},
\frac{K_2}{\beta} \in \naturals$. This assumption can be removed by considering
their integer parts. But we will stick to this assumption to simplify the proof.
Split the $N\times K$ dimensional spreading matrix $\mts$ in to two parts of
dimension $N_1\times K$ and $N_2 \times K$ and denote these matrices by
$\mts_1,\mts_2$ respectively. Let $\mtt_1,\mtt_2$ be two spreading matrices with
dimensions $N_1 \times K_1$ and $N_2 \times K_2$. All the entries of these
matrices are distributed as $\mathcal{N}(0,1)$ and the noise is Gaussian with
variance $\sigma^2$. Similarly split the noise vector $\mbn = (\mbn_1,\mbn_2)$ where 
$\mbn_i$ is of length $N_i$ and $\mbx = (\mbx_1,\mbx_2)$  where $\mbx_i$ is of length $K_i$.
 Let us consider the following Hamiltonian:

\begin{align*}
H_{t,u}(\mbx) =  & -\frac{1}{2\sigma^2} \Vert \mbn_1 + \frac{\sqrt t}{\sqrt{N}}\mts_1
(\mbone-\mbx) + \frac{\sqrt {1-t}}{\sqrt{N_1}}\mtt_1
(\mbone-\mbx_1)\Vert^2 
\\ & -\frac{1}{2\sigma^2} \Vert \mbn_2 + \frac{\sqrt t}{\sqrt{N}}\mts_2
(\mbone-\mbx) + \frac{\sqrt{1-t}}{\sqrt{N_2}}\mtt_2(\mbone-\mbx_2)\Vert^2  + h_u(\mbx)
\end{align*}
Note that the all-one vectors $\mbone$ appearing above are of different dimensions (the dimension
is clear from the context).
For a moment neglect the $h_u(\mbx)$ part of the Hamiltonian and consider the remaining part.  
At $t = 1$, we get the Hamiltonian corresponding to an $N\times K$ CDMA system with 
spreading matrix $\left[ \begin{array}{c}\mts_1 \\ \mts_2\end{array}\right]$. At $t=0$ we get 
the Hamiltonian corresponding to two independent CDMA systems with spreading matrices $\mtt_i$ of dimensions 
$N_i\times K_i$. As before we perturb the Hamiltonian with $h_u(\mbx)$ so that
we can use the concentration results for the magnetization.

Let $Z_{t,u}$ be the partition function with this Hamiltonian and the corresponding average free energy is
given by $g_{t,u} = \frac1K \bE[\ln Z_{t,u}]$. Note that $g_{1,u} = \free_K(u)$ and $g_{0,u} = \frac{K_1}{K}\free_{K_1}(u) + \frac{K_2}{K}\free_{K_2}(u)$. 
From the fundamental theorem of calculus, 
\begin{align}
g_{1,u} = g_{0,u} + \int_0^1 \frac{d}{dt}g_{t,u} dt
\end{align}
Let $\mbz_i = \mbone - \mbx_i$, $\mcuZ_i = \mbn_i + \sqrt{\frac{t}{N}} \mts_i \mbz
+\sqrt\frac{1-t}{N_i} \mtt_i \mbz_i$. Using integration by parts formula with respect to the spreading sequences, the derivative can be simplified as follows
\begin{align}
\frac{d}{dt}g_{t,u} = & \frac{1}{2K\sigma^4}\sum_{i=1,2}\bE\Big\langle \Vert\mcuZ_i \Vert^2 \Big(\frac{1}{N}\Vert \mbz\Vert^2 _-\frac{1}{N_i}\Vert \mbz_i\Vert^2\Big)\Big\rangle_{t,u}\nonumber\\
-&  \frac1{2K\sigma^4}\sum_{i=1,2}\bE\Big\langle \big(\mcuZ_i^{(1)} \cdot \mcuZ_i^{(2)}\big) \big(\frac{1}{N}\mbz^{(1)}\cdot\mbz^{(2)} - \frac{1}{N_i}\mbz_i^{(1)}\cdot\mbz_i^{(2)}\big)\Big\rangle_{t,u}
\end{align}
The system with Hamiltonian $H_{t,u}(\mbx)$ has Nishimori symmetry and hence we can derive results similar to Theorem~\ref{thm:concmag} and Lemma \ref{lem:nishimori1}. In addition to these we need one more Nishimori identity which we did not use before. 
\begin{align}
&\bE\Big\langle \big(\mcuZ_i^{(1)} \cdot \mcuZ_i^{(2)} \big) (\frac{1}{N}\mbz^{(1)}\cdot\mbz^{(2)} - \frac{1}{N_i}\mbz_i^{(1)}\cdot\mbz_i^{(2)}\big)\Big\rangle_{t,u} 	\nonumber\\
&\qquad\qquad\qquad = 
\bE\Big\langle \big(\mbn_i \cdot\mcuZ_i\big) (\frac{1}{N}\mbone\cdot\mbz - \frac{1}{N_i}\mbone\cdot\mbz_i\big)\Big\rangle_{t,u}
\end{align}
Let 
$$m_1 = \frac{1}{K}\sum_{j=1}^K x_j,\;\; m_{11} =
\frac{1}{K_1}\sum_{j=1}^{K_1}x_j,\text{ and }m_{12} =
\frac{1}{K_2}\sum_{j=K_1+1}^{K}x_j$$ 
Let $\epsilon > 0$ be fixed. Using
$\frac1N_i\bE\langle \Vert\mcuZ_i\Vert^2\rangle_{t,u} = 1$ and
Theorem~\ref{thm:concmag}, for a.e., $u>\epsilon$ and a.e., $t>0$, we get 
\begin{align}
\frac{d}{dt}g_{t,u} = &  \frac{\beta}{2K\sigma^4}\sum_{i=1,2}\bE\Big\langle\mbn_i
\cdot \mcuZ_i \Big\rangle_{t,u} \bE\langle m_1 - m_{1i}\rangle_{t,u} + o_K(1)\nonumber\\
=  & \frac{\beta}{2K\sigma^4}\sum_{i=1,2}\bE\Big\langle \mbn_i \cdot\big(\sqrt{\frac{t}{N}} \mts_i \mbz 
+\sqrt\frac{1-t}{N_i} \mtt_i \mbz_i\big) \Big\rangle_{t,u} \bE\langle m_1 - m_{1i}\rangle_{t,u} + o_K(1)
\end{align} 
Now using integration by parts formula with respect to the spreading sequences, and
doing transformations similar to section \ref{sec:transformingT2}, we get for
a.e., $u>\epsilon$ and a.e., $t>0$,
\begin{align}\label{eqn:derivativelimitexistfinal}
\frac{d}{dt}g_{t,u} = &
\frac1{2N\sigma^4}\sum_{i=1,2}\frac{K_i\bE\langle(1-m_1)t +
(1-m_{1i})(1-t)\rangle_{t,u}}{1+\beta\sigma^{-2}\bE\langle(1-m_1)t +
(1-m_{1i})(1-t)\rangle_{t,u}} \bE\langle m_1 - m_{1i}\rangle_{t,u} +
o_K(1)\nonumber \\
= &- \frac1{2K\sigma^2}\sum_{i=1,2}\frac{K_i \bE\langle m_1 -
m_{1i}\rangle_{t,u}}{1+\beta\sigma^{-2}\bE\langle(1-m_1)t +
(1-m_{1i})(1-t)\rangle_{t,u}}+
o_K(1)
\end{align} 
Let us define a function $\eta_{a,b_1,b_2}(t)$ as follows,
\begin{align*}
\eta_{a,b_1,b_2}(t) = - \frac1{2K\sigma^2}\sum_{i=1,2}\frac{K_i (a -
b_i)}{1+\beta\sigma^{-2}((1-a)t +
(1-b_i)(1-t))} 
\end{align*}
Note that for
$a = \bE\langle m_1\rangle_{t,u}$, $b_i = \bE\langle m_{1i}\rangle_{t,u}$ we get
the summation in \eqref{eqn:derivativelimitexistfinal}. When $a,b_i$ satisfy 
\begin{align}\label{eqn:abrelation}
a = \frac{K_1}{K}
b_1 + \frac{K_2}{K} b_2
\end{align}
the function $\eta_{a,b_1,b_2}(t)$ has the following useful properties:
$\eta_{a,b_1,b_2}(1) = 0$ and the derivative with $t$ of this function given by
\begin{align}
\frac1{2K\sigma^4}\sum_{i=1,2}\frac{\beta K_i  (a - b_i)^2}{(1+\beta\sigma^{-2}(1-a)t +
(1-b_{i})(1-t))^2}  \geq 0
\end{align}
Therefore for any $a,b_i$ satisfying \eqref{eqn:abrelation},
$\eta_{a,b_1,b_2}(t) \leq 0$ and hence we can claim the summation in
\eqref{eqn:derivativelimitexistfinal} is also non-positive. 

Bringing the $o_K(1)$ in \eqref{eqn:derivativelimitexistfinal} to the left, we get for a.e., $u > \epsilon$,
\begin{align}
\int_{0}^{1}\frac{d}{dt}g_{t,u} + o_K(1) \leq 0
\end{align} 
Therefore for a.e., $u> \epsilon$, we get
\begin{align}
g_{1,u}  + o_K(1) \leq g_{0,u}
\end{align}
Let $a > \epsilon$ be a constant. Then 
\begin{align*}
\int_{\epsilon}^{a} g_{1,u}du + o_K(1) \leq \int_{\epsilon}^a g_{0,u}du
\end{align*}
which implies
\begin{align*}
\int_{\epsilon}^{a} \free_K(u)du + o_K(1) \leq \frac{K_1}{K}\int_{\epsilon}^a \free_{K_1}(u)du + \frac{K_2}{K}\int_{\epsilon}^a \free_{K_2}(u)du
\end{align*}
which in turn implies that $\lim_{K\to\infty} \int_{\epsilon}^a\free_K(u) du $ exists.
\section{Extensions \label{sec:extensions}}
In this section we briefly describe three variations for which our methods extend in a straightforward manner.
 \subsection{Unequal Powers}
Suppose that the users transmit with unequal powers $P_k$,
$$
y_i= \frac{1}{\sqrt N}\sum_{k=1}^K s_{ik}\sqrt P_k x_k +\sigma n_i
$$
with normalized average power
$\frac1K \sum P_k = 1$. We assume that the empirical distribution of the $P_k$  tends to a distribution and denote the corresponding expectation by $\bE_{ P}[-]$. 

The interpolation method can be applied as before. We interpolate between the true communication system and a decoupled one where 
$$
\tilde y_k=\sqrt P_k x_k +\frac{1}{\sqrt\lambda}w_k
$$
Let ${\cal P}$ denote the diagonal matrix $P_k\delta_{kk^\prime}$. The relevant posterior measure replacing \eqref{eqn:perturbposterior} is now
\begin{align}
p_{t,u}(\mbx\vert\mbn, \mbw,\mbh, \mts)=\frac{1}{Z_{t,u}}
\exp&\biggl(-\frac{1}{2}\Vert \mbn-N^{-{\frac{1}{2}}}B(t)^{\frac{1}{2}}\mts\sqrt {\cal P}(\mbx^0-\mbx)\Vert^2
\\ \nonumber & 
- \frac{1}{2}\Vert \mbw-\lambda(t)^{\frac{1}{2}}\sqrt{\cal P}(\mbx^0-\mbx)\Vert^2 + h_u(\mbx)\biggr) 
\end{align}
where $\lambda(t)$ and $B(t)$ are related as in \eqref{eqn:lambdaBrelation}.
The whole analysis can again be performed in exactly the same manner with the proviso that the 
 correct ``order parameters" are now $m_1 = \frac1N\sum P_k x_k$ and $q_{12} = \frac1N\sum P_kx^{(1)}_kx^{(2)}_k$.  One finds in place of \eqref{eqn:free1}
\begin{align*}
\frac{1}{2}+\bE[f_{1,u}] =& -\frac{1}{2\beta} +\bE_{P}\biggl[\int Dz \ln(2\cosh(\sqrt {P\lambda} z+ P\lambda))\biggr]\\
& -\frac{\lambda}{2}(1+m)-\frac{1}{2\beta}\ln(1+\beta B(1-m)) + \int_0^1 R(t) dt
\end{align*}
where $R(t)$ has the same form as before but the with new definition of $m_1$. 
From the positivity of $R(t)$ we deduce the upper bound \eqref{eqn:upperbound} on the capacity with $c_{RS}(m)$ replaced by
$$
-\bE_{P}\biggl[\int Dz \ln(2\cosh(\sqrt {P\lambda} z+ P\lambda)\biggr]+\frac{\lambda}{2}(1+m) -\frac{1}{2\beta}\ln\lambda\sigma^2
$$

\subsection{Colored Noise}
Now consider the scenario where 
$$
y_i=\frac{1}{\sqrt N}\sum_{k=1}^K s_{ik}x_k + n_i
$$
with colored noise of finite memory. More precisely we assume that the the
covariance matrix $\bE[n_in_j]=C(i,j)$ (depends on $\vert
i-j\vert$) is circulant  as $N\to+\infty$ and has
well defined (real) Fourier transform (the noise spectrum) $\Hat C(\omega)$. The
covariance matrix is real symmetric and thus can be diagonalized by an orthogonal
matrix: $\Gamma=OCO^{T}$ with $OO^{T}=O^{T}O=I$. As $N\to +\infty$ the
eigenvalues are well approximated by $\gamma_n\equiv \Hat C(2\pi\frac{n}{N})$.
Multiplying the received signal by $\Gamma^{-1/2}O$ the input-output relation
becomes
$$
y_i^\prime=\frac{1}{\sqrt N}\sum_{k=1}^K t_{ik} x_k + n_i^\prime
$$
where 
$$
y_i^\prime= (\Gamma^{-1/2}O\mby)_i,\qquad n_i^\prime = (\Gamma^{-1/2}O \mbn)_i
$$
The new noise vector $\mbn^\prime$ is white with unit variance, but the spreading matrix is now correlated with
\begin{equation}\label{cov}
\bE[t_{ik} t_{jl}]=\delta_{ij}\delta_{kl}\gamma_i^{-1}
\end{equation}
One may guess that
this time the interpolation is done between the true system and the decoupled channels 
$$
\tilde y_k= x_k+ \frac{1}{\sqrt\lambda_{col}} w_k
$$
where this time 
$$
\lambda_{col}= \int_0^{2\pi}\frac{d\omega}{2\pi}\frac{B}{\Hat C(\omega)+\beta B(1-m)}
$$
Note that $\Hat C(\omega) = 1$ when the noise is white and we get back the
$\lambda$ defined in \eqref{lambda}.
The interpolating system has the same posterior as in \eqref{eqn:perturbposterior}
 but with $\lambda_{col}(t)$ and $B(t)$ related by
$$
\int_0^{2\pi} \frac{d\omega}{2\pi}\frac{B(t)}{\Hat C(\omega)+\beta B(t) (1-m)} + \lambda_{col}(t)
=\int_0^{2\pi} \frac{d\omega}{2\pi}\frac{B}{\Hat C(\omega) + \beta B(1-m)}
$$
The only difference in the subsequent analysis is in the algebraic manipulations
for the term $T_2$ in section \ref{sec:transformingT2}. Indeed these require
integrations by parts with respect to the spreading sequence which involve
\eqref{cov}. The analog of \eqref{eqn:T2} now becomes 
\begin{align}
T_2 &= \frac1N\sum_{n=1}^N\frac{B'(t)\bE\langle 1- m_1\rangle_t}{2(\gamma_n + \beta B(t)\bE\langle1-m_1\rangle_t)}\\
\nonumber
&\to 
\int_0^{2\pi}\frac{d\omega}{2\pi} \frac{B'(t)\bE\langle 1-
m_1\rangle_t}{2(S(\omega) + \beta B(t)\bE\langle1-m_1\rangle_t)}
d \omega
\nonumber
\end{align}
This finally leads to the bound on capacity with $c_{RS}(m)$ replaced by,
\begin{align}\nonumber
 -\int Dz\ln 2\cosh (\sqrt {\lambda_{col}} z + \lambda_{col}) & +
\frac{\lambda_{col}}{2}(1+m)
\\ \nonumber &
+\frac{1}{2\beta}\int_0^{2\pi} \frac{d\omega}{2\pi}\ln\frac{\Hat C(\omega)}{\Hat C(\omega) + \beta (1-m)}
\nonumber
\end{align}


\subsection{Gaussian Input}\label{sec:gaussianinter}

The interpolation method also works for non binary inputs. Here we consider the
simplest case of Gaussian inputs with distribution \eqref{eqn:gaussianprior}
(which achieves the maximum of the mutual information for any symmetric
$s_{ik}$). Here we outline the necessary changes in the analysis.

The interpolation is done as explained in section \ref{sec:upperbound} except
that \eqref{eqn:postdistt} is multiplied by the Gaussian distribution
\eqref{eqn:gaussianprior}. In \eqref{eqn:distyyt} we also have to include this
Gaussian factor and the sum over $\mbx_0$ is replaced by an integral. Then as in
section \ref{sec:preliminaries} we do the change of variables 
$\mby\to B(t)^{-1/2} +N^{-1/2} \mts\mbx^0$ and $\tilde \mby\to \lambda(t)^{-1/2}\mbw+ \mbx^0$.
 The posterior measure used for the interpolation becomes 
\begin{align}
p_{t,u}(\mbx\vert\mbn, \mbw,\mbh, \mts)=\frac{1}{Z_{t,u}}
\exp&\biggl(-\frac{1}{2}\Vert \mbn-N^{-{\frac{1}{2}}}B(t)^{\frac{1}{2}}\mts(\mbx^0-\mbx)\Vert^2
\\ \nonumber & 
- \frac{1}{2}\Vert \mbw-\lambda(t)^{\frac{1}{2}}(\mbx^0-\mbx)\Vert^2 + h_u(\mbx)\biggr) \frac{e^{-\frac{\vert\vert\mbx\vert\vert^2}{2}}}{(2\pi)^{N/2}}
\end{align}
and we have to compute $\lim_{K\to+\infty}\lim_{u\to 0} \bE[f_{1,u}(\mbn,\mbw,\mbh,\mbs,\mbx^0)]$. The main difference is that now the 
expectation $\bE$ is also with respect to the Gaussian vector $\mbx_0$. The
algebra is done as in section \ref{sec:upperboundproof} except that $x_k^0$ is
not set to one, $z_k$ is replaced by $x_k^0-x_k$ and the correct order
parameters are 
 $m_1=\frac1K\sum x^0_{k}x_k$ and 
$q_{12} = \frac1K \sum x_k^{(1)}x_k^{(2)}$.

The interpolation method then yields in place of \eqref{eqn:free1}
\begin{align*}
\frac{1}{2}+\bE[f_{1,u}] &= -\frac{1}{2\beta} -\frac12 \ln(1+\lambda) -
\frac{1}{2\beta}\ln(1+\beta B(1-m)) \\&+ \frac{\lambda}{2}(1-m) + \int_0^1 R(t) dt + O(\sqrt u)
\end{align*}
where $R(t)$ is the same function as before but with new definition of $m_1$.  Again the positivity of $R(t)$ implies that the replica solution is an upper bound to the capacity.

\section{Concluding remarks}
In this contribution we have shown that the capacity of binary input CDMA system
with random spreading 
is upper bounded by the formula conjectured by Tanaka using replica method. 
The approach we follow is by developing an interpolation method for this system.
This idea has its origins in statistical mechanics and has been applied to
Gaussian energy models. The current system is very much different from those
models and the
proof we develop is also significantly different. In fact this model is closer
to the Hopfield model for neural networks, for which the interpolation method is
still an open problem. 

We also show that the capacity and the free energy functions concentrate around
their average in the large-system limit. In addition we prove a weak
concentration for the magnetization for a system which is slightly perturbed
using a Gaussian field. It might be interesting to show a similar result for the
CDMA system itself which has some implications towards proving the
concentration of the BER. 
We also show the independence of the capacity from the spreading sequence
distributions in the large-system limit.


We expect that the powerful probabilistic tools used here have applications for other similar
situations in communication systems. We have shown some of the extensions here
but there are many other cases like constellations other
than binary, CDMA with LDPC coded communication to name a few, to which this method can be
applied. In all these cases we can prove an upper bound on the capacity.
The most interesting and also important open problem is to prove the lower bound. This
seems to be a difficult problem and again the standard techniques fail. 
Other important problems are proving the conjectures related to the BER of
various decoders.  


\begin{appendix}\label{appendix}

\section{Relation between capacity and free energy}
\label{apen:capfree}
Replacing (\ref{eqn:postdist}) in the conditional entropy
\begin{align}
H(\mbX\vert \mbY)= & -\bE_{\mbY}\biggl[\sum_{\mbx} p(\mbx\vert \mby,\mts)\ln p(\mbx\vert \mby,\mts)\biggr]\nonumber
\\&
= \bE_{\mbY}\biggl[\sum_{\mbx} p(\mbx\vert \mby,\mts)\ln Z(\mby,\mts)\biggr]
+ \bE_{\mbY}\biggl[\sum_{\mbx} p(\mbx\vert \mby,\mts)\frac{1}{2\sigma^2}
\Vert N^{-\frac{1}{2}}\mts\mbx-\mby\Vert^2\biggr]\\ \nonumber
&
-\bE_{\mbY}[\sum_{\mbx}p(\mbx\mid\mby,\mts)\ln p_{\mbX}(\mbx)]
\end{align}
The first term on the r.h.s is equal to $\bE_{\mbY}[\ln Z(\mby,\mts)]$ because
$\sum_{\mbx}p(\mbx\mid \mby,\mts)=1$.
The second term on the r.h.s can be computed exactly. Indeed,
\begin{align}
&\bE_{\mbY}\biggl[\sum_{\mbx} p(\mbx\vert \mby,\mts)\frac{1}{2\sigma^2}\Vert N^{-\frac{1}{2}}\mts\mbx-\mby\Vert^2\biggr]\nonumber
\\&
=\int d\mby \frac{Z(\mby, \mts)}{(\sqrt{2\pi\sigma^2})^N}\sum_{\mbx} p(\mbx\vert \mby,\mts)\frac{1}{2\sigma^2}\Vert N^{-\frac{1}{2}}\mts\mbx-\mby\Vert^2
\nonumber
\\&
= \sum_{\mbx}p_{\mbX}(\mbx)\int d\mby
\frac{1}{(\sqrt{2\pi\sigma^2})^N}e^{-\frac{1}{2\sigma^2}\Vert
N^{-\frac{1}{2}}\mts\mbx-\mby\Vert^2} \nonumber\\
&\times \frac{1}{2\sigma^2}\Vert N^{-\frac{1}{2}}\mts\mbx-\mby\Vert^2
\nonumber
\\&
= \frac{N}{2} =  \frac{K}{2\beta}\nonumber
\end{align}
A similar calculation shows that the third term is equal to $H(\mbX)$.
Therefore the relation between Shannon's conditional entropy and the free energy is
\begin{equation*}
H(\mbX\mid\mbY)=\bE_{\mbY}[\ln Z(\mby,\mts)] + \frac{K}{2\beta} + H(\mbX)
\end{equation*}
This is equivalent to the announced relation \eqref{eqn:capfreerel}.

\section{Probabilistic tools}\label{apen:probtools}

Our proofs rely on a general concentration theorem for suitable Lipschitz functions of many Gaussian random variables \cite{Tal96}, \cite{Tal03} and this is why we need Gaussian signature sequences. 
In the version that we use here we need functions that are Lipschitz with
respect to the Euclidean distance. More precisely we say that a function $f:
\mathbb{R}^M\rightarrow \mathbb{R}$ is a Lipschitz function with  constant $L_M$
if for all $(\mbu, \mbv)\in \mathbb{R}^M\times \mathbb{R}^M$
\begin{equation*}
 \vert f(\mbu) - f(\mbv)\vert\leq L_M \Vert \mbu-\mbv\Vert
\end{equation*}
When another distance is used the function will still be Lipschitz but one has to carefully keep track of the possibly qualitatively different $M$ dependence.

\vskip 0.25cm

\begin{theorem}{\cite{Tal96}}\label{gauss_conc}
\noindent Let $(U_1,...,U_M)={\underline U}$ be $M$ independent identically distributed Gaussian random variables with distribution $\mathcal{N}(0,v^2)$ and 
let $f: \mathbb{R}^M\rightarrow \mathbb{R}$ be Lipschitz with respect to the Euclidean distance, with constant $L_M$. Then $f$ satisfies
\begin{equation*}
\bP[\vert f(\underline u) - \bE [f(\underline u)]\vert \geq t ] \leq 2e^{-\frac{t^2}{2v^2 L_M^2}}
\end{equation*}
\end{theorem}

\vskip 0.25cm

In our application it will not be possible to apply directly this theorem because the relevant function is Lipschitz only on a subset
$G\subset\mathbb{R}^M$. 
It turns out that the measure of the complement $G^c$ is negligible as $M\to +\infty$. For the ``good part'' of the function supported on $G$ we will use the following
result of McShane and Whitney

\vskip 0.25cm

\begin{theorem}{\cite{Hei05}}\label{shane-whitney}
\noindent Let $f:G \rightarrow \mathbb R$, be Lipschitz over $G \subset \mathbb R^M$ with constant $L_M$. Then there exists an 
extension $g: \mathbb R^M\rightarrow \mathbb R$ such that $g\vert_{G} = f$ which is Lipschitz with the same constant over the whole of $\mathbb{R}^M$.
\end{theorem}


From these two theorems we can prove the following.

\vskip 0.25cm

\begin{lemma}\label{littlelemma}
 
\noindent Let $f$ and $g$ be as in theorem \ref{shane-whitney}. Assume $0\in G$
and $\bE[f(\mbu)^2]\leq C^2$, $f(0)^2 \leq C^2$ for some positive number $C$.  Then for 
\begin{equation*}
\frac{t}{2}\geq 3(C + v\sqrt{M}) \sqrt{\mathbb{P}(G^c)}
\end{equation*}
we have
\begin{equation*}
\mathbb{P}[\vert f(\mbu) - \bE [f(\mbu)] \vert \geq t]\leq 2e^{-\frac{t^2}{8v^2 L_M^2}} + P[G^c]
\end{equation*} 
\end{lemma}
\vskip 0.25cm
\begin{proof}
We drop the $\mbu$ dependence to lighten the notation. Notice that $0\in G$ implies $f(0)=g(0)$. Thus $g(0)^2\leq C^2$.
Also, since $g$ is Lipschitz on the whole of $\mathbb{R}^M$
\begin{align}
\bE[g^2]&\leq 2(g(0)^2+ \bE[(g-g(0))^2])\nonumber
\\&
\leq 2(C^2+L_M\bE[\Vert \mbu^2\Vert)\nonumber
\\&=2(C^2+M v^2L_M)\nonumber
\end{align}
Furthermore on $G$ we have $g=f$, so by the Cauchy-Schwartz inequality
\begin{align}
\vert\bE[g-f]\vert&= \vert\bE[(g-f)1_{G^c}]\vert\nonumber
\\&
\leq (\bE[g^2]^{1/2}+\bE[f^2]^{1/2})\sqrt{\bP[G^c]}
\nonumber
\\&
\leq (C + \sqrt{2}(C^2+M v^2L_M)^{1/2}) \sqrt{\mathbb{P}[G^c]}
\nonumber 
\\&
\leq 3(C+ v\sqrt{M}L_M)\sqrt{\mathbb{P}[G^c]}
\leq
\frac{t}{2}\nonumber
\end{align}
Moreover
\begin{align}
\mathbb{P}[\vert f - \bE f \vert \geq t]  &= \mathbb{P}[\vert g - \bE f\vert \geq t\mid \mbU\in G]\mathbb{P}[G]
\nonumber
\\&\,\,\,\,\,\,
+ \bP[\vert f - \bE f\vert \geq t\mid \mbU\in G^c]\mathbb{P}[G^c]\nonumber
\\&
\leq \mathbb{P}[\vert g - \bE g\vert \geq t - \vert \bE g - \bE f \vert] + \mathbb{P}[G^c]\nonumber
\end{align}
The result of the lemma then follows from
\begin{equation*}
\mathbb{P}[\vert g - \bE g\vert \geq t - \vert \bE g - \bE f \vert]\leq \mathbb{P}[\vert g - \bE g\vert \geq \frac{t}{2}]
\end{equation*}
and the application of theorem \ref{gauss_conc}.
\end{proof}
\vskip 0.25cm
In order to prove Theorems \ref{thm:capconc} and \ref{thm:freeconc} it will be sufficient to find  suitable sets $G$ with measure
nearly equal to one (as $M\to+\infty$), on which the capacity and free energy have a Lipschitz constant $L_M\to 0$.

\section{Proofs of Theorems \ref{thm:capconc} and
\ref{thm:freeconc}}\label{apen:concproofs}

For the proofs, it is convenient to reformulate the statements of the theorems as follows. Let $\underline{1}$ be the $K$ dimensional vector 
$(1,...,1)$, $\mts^0$ be the $K\times N$ matrix with elements $s_{ik}x_k^0$. We
set $p_{\mbX}^0(\mbx)=\prod_{k=1}^{K} p_X(x_kx_k^0)$ and consider the partition
function
\begin{equation}\label{partfct2}
Z^\prime(\mbn, \mts^0) = \sum_{\mbx} p_{\mbX}^0 (\mbx)e^{-\frac{1}{2\sigma^2}\Vert N^{-1/2} \mts^0(\mbx-\underline{1}) - \sigma\mbn\Vert^2}
\end{equation}
where we recall that $\mbn=(n_1,...,n_N)$ are independent Gaussian variables $\mathcal{N}(0,1)$.
Notice that due to the invariance of the distribution of $s_{ik}$ under the transformation $s_{ik}\to x_{k}^0 s_{ik}$,
\begin{equation*}
\bE_{\mbN,\mtS}[\ln Z^\prime(\mbn,\mts^0)]=\bE_{\mbN,\mtS}[\ln Z^\prime(\mbn,\mts)]
\end{equation*}
The statements of Theorems \ref{thm:capconc} and \ref{thm:freeconc} are equivalent to 
\begin{align}\label{reform1}
\bP[\vert \sum_{\mbx^0}p_{\mbX}(\mbx^0)\bE_{\mbN}[\ln Z^\prime(\mbn, \mts^0)] -
&\bE_{\mbN,\mtS}[\ln Z^\prime(\mbn,\mts)]\vert\geq tK] \leq 3 e^{-\alpha_1
t^2{N}}
\end{align}
and 
\begin{equation}\label{reform2}
\bP[\sum_{\mbx^0}p_{\mbX}(\mbx^0)\vert \ln Z^\prime(\mbn, \mts^0) -
\bE_{\mbN,\mtS}[\ln Z^\prime(\mbn,\mts)]\vert\geq tK]\leq 3 e^{-\alpha_2
t^2\sqrt{N}}
\end{equation}
To see this use the  
change of variable $\mby= N^{-1/2}\mts\mbx_0 +\sigma\mbn$ followed by $x_k\to
x_kx_{k}^0$ in the partition function summation \eqref{eqn:partition}.

\subsection{Proof of (\ref{reform1})}

Let $B$ be a positive constant to be chosen later and define
\begin{equation*}
G=\{\mts\mid {\rm for\,\,all\,\,} \mbx,\mbx^0, \Vert \mts^0(\mbx-\underline{1})\Vert^2\leq BN\}
\end{equation*}
\vskip 0.25cm
\begin{lemma}\label{measbadset}

\noindent We have the following estimate for the measure of $G^c$,
\begin{equation*}
\bP(G^c)\leq 3^{K}2^{\frac{N}{2}} e^{-\frac{B}{16\beta}}
\end{equation*}
\end{lemma}
\vskip 0.25cm

\begin{proof} 
First notice that for any given $\mbx$, 
\begin{equation*}
\frac{1}{\sqrt K}\sum_{k=1}^K s_{ik}^0(x_k-1),\qquad i=1,...,N
\end{equation*}
are independent Gaussian random variables with zero mean and variance $(a^2)$ smaller than $4$.
Thus the identity
\begin{equation*}
\int dx \frac{e^{-\frac{x^2}{2a^2}}}{\sqrt{2\pi a^2}} e^{\frac {x^2}{16}}=\bigl(1-\frac{a^2}{8}\bigr)^{-\frac{1}{2}}
\end{equation*}
implies (because $a^2\leq 4$)
\begin{equation*}
\bE[e^{\frac{1}{16 K}\Vert\mts^0(\mbx-\underline{1})\Vert^2}]\leq 2^{\frac{N}{2}}
\end{equation*}
Then from the Markov inequality, for any $\mbx$
\begin{equation*}
\bP(\Vert\mts^0(\mbx-\underline{1})\Vert^2\geq BN)\leq 2^{\frac{N}{2}} e^{-\frac{BN}{16K}}=2^{\frac{N}{2}} e^{-\frac{B}{16\beta}}
\end{equation*}
The result of the lemma then follows from the union bound over $3^K$ possible
$\mbx_0 - \mbx$ vectors.
\end{proof}
\vskip 0.25cm

We will apply Lemma \ref{measbadset} to 
\begin{align*}
f(\mts)=\frac1{K}\sum_{\mbx^0}p_{\underline X}(\underline x^0)\bE_{\mbN}[\ln Z^\prime(\mbn,\mts^0)]
\end{align*}
for a suitable choice of $B$. In the application the matrix $\mts$ is to be thought as a vector with $KN$ components and norm 
\begin{equation*}
\Vert \mts\Vert = \biggl(\sum_{i=1}^N\sum_{k=1}^{K} s_{ik}^2\biggr)^{\frac{1}{2}}
\end{equation*}

Clearly $0\in G$ and $f(0)^2 = (\frac1K \bE_{\mbN}[\frac{1}{2}\Vert \mbn
\Vert^2])^2 = 1/4\beta^2$. Also it is evident that 
$\ln Z^\prime(\mbn,\mts^0)\leq 0$. On the other hand restricting the sum in the partition function to $\mbx=1$ we have
\begin{equation*}
\frac{1}{K}\sum_{\mbx^0}p_{\mbX}(\mbx^0)\bE_{\mbN}[\ln Z^\prime(\mbn,\mts^0)]\geq -\frac{1}{2\sigma^2K}\bE_{\mbN}[\sigma^2\Vert \mbn\Vert^2]-\frac{1}{K} H(\mbX)\geq -\frac{N}{2K}-\ln 2
\end{equation*}
Therefore we have 
\begin{equation*}
\bE_{\mtS}[f(\mts)^2]\leq (\frac{1}{2\beta}+\ln 2)^2
\end{equation*} 

Let us now compute the Lipschitz constant.
\vskip 0.25cm
\begin{lemma}\label{lipconst1}
 $K^{-1}\bE_{\mbN}[\sum_{\mbx^0}p_{\mbX}(\mbx^0)\ln Z^\prime(\mbn,\mts^0)]$ is Lipschitz on $G$, with constant 
\begin{equation*}
L_N=\sigma^{-2}\sqrt\beta K^{-1}(\sqrt{B}+\sqrt{N}\sigma)
\end{equation*}
\end{lemma}
\vskip 0.25cm
\begin{proof}
The exponent of the partition function is\footnote{ a Hamiltonian}
\begin{equation}\label{ham}
H(\mbn,\mts^0, \mbx)=\frac{1}{2\sigma^2}\Vert N^{-1/2}\mts^0(\mbx-\underline{1})-\sigma\mbn\Vert^2
\end{equation}
In the section \ref{appen:lipschitzH1} we show that for $(\mts,\mtt)\in G\times G$

\begin{equation}\label{inequham1}
\vert H(\mbn,\mts^0, \mbx)- H(\mbn,\mtt^0, \mbx)\vert\leq \sigma^{-2}2\sqrt
\beta(\sqrt B+\Vert n\Vert)\Vert \mts-\mtt\Vert
\end{equation}
Using this inequality together with
\begin{equation*}
- H(\mbn,\mts^0, \mbx)\leq - H(\mbn,\mtt^0, \mbx) + 
\vert H(\mbn,\mts^0, \mbx)- H(\mbn,\mtt^0, \mbx)\vert
\end{equation*}
we have for $(\mts,\mtt)\in G\times G$
\begin{align}
&\ln\frac{\sum_{\mbx}p_{\mbX}^0(\mbx)\exp(- H(\mbn,\mts^0,
\mbx))}{\sum_{\mbx}p_{\mbX}^0(\mbx)\exp(- H(\mbn,\mtt^0, \mbx))}\nonumber
\\&
\leq\ln\frac{\sum_{\mbx}p_{\mbX}^0(\mbx)\exp(\vert H(\mbn,\mts^0, \mbx)- H(\mbn,\mtt^0, \mbx)\vert - H(\mbn,\mtt^0, \mbx))
}{\sum_{\mbx}p_{\mbX}^0(\mbx)\exp(- H(\mbn,\mtt^0, \mbx))}
\nonumber
\\&
\leq
\sigma^{-2}2\sqrt \beta(\sqrt B+\Vert n\Vert)\Vert \mts-\mtt\Vert
\nonumber
\end{align}
Therefore taking the expectation over the noise,  we get
\begin{align}
\vert\sum_{\mbx^0} p_{\mbX}(\mbx^0)\bE_{\mbN}&[\ln Z^\prime(\mbn,\mts^0)]- \sum_{\mbx^0}p_{\mbX}(\mbx^0)\bE_{\mbN}[\ln Z^\prime(\mbn,\mtt^0)]\vert\nonumber
\\&
\leq \sigma^{-2}2\sqrt\beta(\sqrt B+\sigma\bE[\Vert \mbn\Vert])\Vert \mts-\mtt\Vert
\nonumber
\\&
\leq \sigma^{-2}2\sqrt\beta(\sqrt B+\sigma\bE[\Vert \mbn\Vert^2]^{1/2})\Vert \mts-\mtt\Vert
\nonumber
\end{align}
which yields the Lipschitz constant of the lemma.
\end{proof}
\vskip 0.25cm

Finally (\ref{reform1}) follows from Lemmas \ref{littlelemma}, \ref{measbadset} and \ref{lipconst1} with the choice $B= 32\beta(2K+N)$. 
We obtain $\alpha_1= 1/(8KL_N^2) \geq \sigma^{4}/(16\beta(64\beta+32+\sigma^2))$. 

\subsection{Proof of (\ref{reform2})}

This case is more cumbersome but the ideas are the same. We choose the set $G$ as 
\begin{equation*}
G = \Bigl\lbrace \mts,\mbn \mid \max_i\vert n_i\vert \leq \sqrt A\,\, {\rm and\,\, for\,\, all\,\,} \mbx,    \Vert\mts^0(\mbx-\underline{1})\Vert^2\leq BN\Bigr\rbrace
\end{equation*}
where, as before $A$ and $B$ will be chosen appropriately later on. For Gaussian noise 
$\bP[\vert n_i\vert\geq \sqrt A]\leq 4 e^{-\frac{A}{4}}$ therefore from the union bound
$\bP(\max_i\vert \mbn_i\vert\geq \sqrt A)\leq 4 Ne^{-\frac{A}{4}}$. Using Lemma \ref{measbadset} we obtain an estimate
for the measure of $G^c$,
\begin{equation*}
\bP[G^c]\leq 4N e^{-\frac{A}{4}}+2^{K+\frac{N}{2}}e^{-\frac{B}{16\beta}}
\end{equation*}

The goal is to apply Lemma \ref{littlelemma} to $f(\mbn,\mts)=\ln Z^\prime(\mbn, \mts^0)$ defined on 
$\mathbb{R}^K\times \mathbb{R}^{NK}$. 

Clearly $(0,0)\in G$, $f(0,0) = \ln 2$ and by the same argument as before we have
$\bE[f(\mbn,\mts)^2]\leq (\frac{1}{2\beta}+ \ln 2)^2=C^2$. 
It remains to compute the Lipschitz constant.
\vskip 0.25cm
\begin{lemma}\label{lipconst2}
The free energy $K^{-1}\ln Z^\prime(\mbn,\mts^0)$ is Lipschitz on $G$ with constant 
\begin{equation*}
L_N=\sigma^{-2}(2\sqrt\beta+\sigma)K^{-1}(\sigma\sqrt{NA}+\sqrt{B})
\end{equation*}
\end{lemma}
\vskip 0.25cm
\begin{proof}
For the same Hamiltonian (\ref{ham}) we show in section \ref{appen:lipschitzH2}
\begin{align}
&\vert H(\mbn,\mbs^0,\mbx)-H(\mbn,\mtt^0, \mbx)\vert\nonumber
\\&
\leq \sigma^{-2}2(2\sqrt\beta+\sigma)(\sigma\sqrt{NA}+\sqrt{B})\Vert(\mbn,\mts)-(\mbm,\mtt)\Vert\label{inequham2}
\end{align}
Then proceeding in the same way as in the proof of Lemma \ref{lipconst1} we get
\begin{align}
&\vert\ln Z^\prime(\mbn,\mts^0)-\ln Z^\prime(\mbm,\mtt^0)\vert \nonumber
\\&
\leq \sigma^{-2}(2\sqrt\beta+\sigma)(\sigma\sqrt{NA}+\sqrt{B})\Vert(\mbn,\mts)-(\mbm,\mtt)\Vert\nonumber
\end{align}
\end{proof}

We can now conclude the proof of (\ref{reform2}) by collecting the previous results and choosing
$A=\sqrt{N}/\sigma^2$ and $B=32\beta (K+N)$. This gives $\alpha_2 = 1/(8\sqrt K
L_N^2) \geq \sigma^4\beta^\frac32/(32(2\sqrt\beta+\sigma)^2)$.

\subsection{Proof of (\ref{inequham2})}\label{appen:lipschitzH2}

Let $\mbn$, $\mbm$ be two noise realizations and $\mts$, $\mtt$ two spreading sequences all belonging to the appropriate set
$G$. Let $\mby = \mbx - \underline{1}$. First we expand the Euclidean norms
\begin{align}
&\Vert N^{-\frac{1}{2}}\mts^0 \mby - \sigma\mbn\Vert^2- \Vert N^{-\frac{1}{2}}\mtt^0 \mby - \sigma\mbm\Vert^2
\nonumber
\\&
=\sigma^2\Vert \mbn\Vert^2-\sigma^2\Vert\mbm\Vert^2+N^{-1}(\Vert \mts^0 \mby\Vert^2-\Vert\mtt^0 \mby\Vert^2)
\nonumber
\\&
\,\,\,\,\,\, - 2\sigma N^{-\frac{1}{2}}(
\mbn^t\cdot\mts^0\mby - \mbm^t\cdot\mtt^0\mby)
\nonumber
\\&
=\sigma^2(\mbn-\mbm)^t\cdot(\mbn+\mbm) +N^{-1}(\mts^0 \mby-\mtt^0 \mby)^t\cdot(\mts^0 \mby+\mtt^0 \mby)
\nonumber
\\&
\,\,\,\,\,\,-2\sigma N^{-\frac{1}{2}}(\mbn-\mbm)^t\cdot\mts^0\mby - 2\sigma N^{-\frac{1}{2}} \mbm^t\cdot(\mts^0\mby-\mtt^0\mby)
\nonumber
\end{align}
We estimate each of the four terms on the right hand side of the last equality. By Cauchy-Schwartz the first term is bounded by
\begin{align}
\Vert \mbn-\mbm\Vert \Vert \mbn+\mbm\Vert&
\leq \sqrt{N}{\rm max}_i (\vert n_i\vert+\vert m_i\vert)\Vert \mbn-\mbm\Vert
\nonumber
\\&
\leq  2\sqrt{NA}\Vert \mbn-\mbm\Vert
\nonumber
\end{align}
Using Cauchy-Schwartz and $\Vert (\mts^0-\mtt^0)\mby\Vert\leq \Vert\mts^0-\mtt^0\Vert\Vert \mby\Vert$ where $\Vert\mts^0-\mtt^0\Vert=\Vert\mts-\mtt\Vert$ is the (Hilbert-Schmidt) norm,
\begin{equation*}
\Vert \mts-\mtt\Vert=\biggl(\sum_{i=1}^N\sum_{l=1}^K (s_{il}-t_{il})^2\biggr)^{1/2}
\end{equation*}
we obtain for the second term the estimate
\begin{align}
N^{-1}\Vert\mts  -\mtt\Vert\Vert \mby\Vert(\Vert\mts^0\mby\Vert +  \Vert\mtt^0\mby \Vert)
& \leq 
N^{-1}\Vert\mts -\mtt\Vert 2\sqrt{K}2\sqrt{BN}
\nonumber
\\&
=4\sqrt{\beta B}\Vert\mts -\mtt\Vert
\nonumber
\end{align}
Similarly the third term is bounded by,
\begin{align}
2N^{-\frac{1}{2}}\Vert \mbn-\mbm\Vert\Vert \mts^0\mby\Vert
&\leq 2N^{-\frac{1}{2}}\Vert \mbn-\mbm\Vert\sqrt{BN}
\nonumber
\\&
=2\sqrt{B}\Vert \mbn-\mbm\Vert
\nonumber
\end{align}
and the fourth one by
\begin{align}
2N^{-\frac{1}{2}}\Vert \mbm\Vert \Vert \mts-\mtt\Vert\Vert \mby\Vert
&\leq
2N^{-\frac{1}{2}}\sqrt{NA} \Vert \mts-\mtt\Vert2\sqrt{K}
\nonumber
\\&
= 4\sqrt{\beta NA}\Vert \mts-\mtt\Vert
\nonumber
\end{align}
Collecting all four estimates we obtain
\begin{align}
&\Vert N^{-\frac{1}{2}}\mts^0 (\mbx-\underline{1}) - \sigma\mbn\Vert^2- \Vert N^{-\frac{1}{2}}\mtt^0 (\mbx-\underline{1}) - \sigma\mbm\Vert^2
\nonumber
\\&
\leq 2\sigma(\sigma\sqrt{NA}+\sqrt{B})\Vert \mbn-\mbm\Vert + 4\sqrt{\beta}(\sigma\sqrt{ NA}+\sqrt{ B})\Vert \mts-\mtt\Vert
\nonumber
\\&
\leq 
2(2\sqrt{\beta}+\sigma)(\sigma\sqrt{NA}+\sqrt{B})\Vert(\mbn,\mts)-(\mbm,\mtt)\Vert
\nonumber
\end{align}
where the last norm is the Euclidean norm in $\mathbb{R}^N\times \mathbb{R}^{NK}$.

\subsection{Proof of (\ref{inequham1})}\label{appen:lipschitzH1}

Let $\mts$ and $\mtt$ be two spreading sequences both belonging to the appropriate $G$. Let $\mby = \mbx - \underline{1}$. Following similar steps 
as in the previous paragraph with $\mbn=\mbm$ the result can be read off
\begin{align}
\Vert N^{-\frac{1}{2}}\mts^0 \mby - \sigma\mbn\Vert^2-
 \Vert N^{-\frac{1}{2}}\mtt^0 \mby & - \sigma\mbn\Vert^2
\nonumber
\\&
\leq 4\sqrt{\beta}(\sqrt{B}+\sigma\Vert \mbn\Vert)\Vert \mts-\mtt\Vert
\nonumber
\end{align}

\section{Proof of Theorem \ref{thm:capconcbin}}
The idea of this proof is based on \cite{PaS91},\cite{ShT93}. 
\begin{proof}
Here, for simplicity of notation and without loss of generality, we assume the noise variance to be $1$ and the second and
fourth moments of spreading sequences to be less than $1$.
For $l\leq K$, let $\phi_l$ be the sigma algebra generated by $\{s_{ik}: 1\leq i \leq N,1
\leq k\leq l\}$.
and set
\begin{align*}
f_l = \bE\left[I(\mbX;\mbY)|\phi_l\right],\quad \psi_l = f_l - f_{l-1}
\end{align*}
Then 
\begin{align*}
\bE(I(\mbX;\mbY) - \bE[I(\mbX;\mbY)])^2 & = \sum_{l=1}^{K}\bE[\psi_l^2]
\end{align*}
The goal is to bound each term in this sum by $O(\frac1{K^2})$.
Here we use the following form of the mutual information
\begin{align*}
I(\mbX;\mbY) = - \frac1{2\beta} - \bE_{\mbN}\Big[\sum_{\mbx^0}p_{\mbX}(\mbx^0)\ln \sum_{\mbx}p_{\mbX}(\mbx)e^{H(\mbx_0,\mbx)}\Big]
\end{align*}
where,
\begin{align*}
H(\mbx_0,\mbx) & = -\frac{1}{2}\sum_i\Big(n_i+\frac1{\sqrt N}\sum_k
s_{ik}(x^0_{k}-x_k)\Big)^2 \\
& = -\frac12 \sum_i n_i^2 - \frac{1}{\sqrt N }\sum_{i,k} n_i 
s_{ik}x^0_k -\frac1{2N} \sum_i \Big(\sum_k s_{ik}(x^0_{k}-x_k)\Big)^2 +
\frac{1}{\sqrt N }\sum_{ik} n_i s_{ik}x_k
\end{align*}
In the above expanded form, the first two terms do not involve $\mbx$ and hence
the concentration of these terms follows very easily. Therefore, in the rest of the proof
we consider the Hamiltonian with only the remaining two terms. From now on
in the notation, we do not explicitly show the dependency of $H$ on $\mbx^0$ and
$\mbx$. To this end we define the following three Hamiltonians.
\begin{align*}
H_l &= \frac{-1}{2N}\sum_{k_1,k_2 \neq l,i}s_{ik_1}s_{ik_2}
(x^0_{k_1}-x_{k_1})(x^0_{k_2}-x_{k_2}) \\
&+ \frac{1}{\sqrt N} \sum_{i, k\neq l} n_is_{ik}x_k\\
R_l &= \frac{1}{2N}\sum_is_{il}^2(x^0_l-x_l)^2\\&-\frac{1}{N}\sum_{i,k}
s_{ik}s_{il}(x^0_l-x_l)(x^0_{k}-x_k)+\frac{1}{\sqrt N}\sum_in_is_{il} x_l\\
\tilde H_l(t)  & = H_l + t R_l
\end{align*}
where $t\in [0,1]$ will play the role of an interpolating parameter.
We also introduce the difference of free energies associated to the Hamiltonian $\tilde H_l(t)$ and $H_l$,
\begin{equation*}
\tilde{f}_l(t) = \sum_{\mbx^0}p_{\mbX}(\mbx^0)(\ln Z(\tilde H_l(t)) - \ln
Z(\tilde H_l(0)))
\end{equation*}
In the last definition the partition function is defined by the usual summation over all configurations $\mbx$.

With these definitions we have the representation
\begin{equation*}
\psi_l = \frac1K\bE_{\geq l+1}\tilde{f}_l(1) -\frac1K\bE_{\geq l} \tilde{f}_l(1)
\end{equation*}
where $\bE_{\geq l}$ means expectation with respect to $\{s_{ik}\; \forall\; k \geq
l\}$.
Using convexity in the form of $\bE_{\geq l+1}[\tilde{f}_l(1)]^2 \leq \bE_{\geq
l+1}[\tilde{f}_l(1)^2]$, it follows that
\begin{align*}
\bE[\psi_l^2] 
& \leq \frac{1}{K^2}\bE \bE_{\geq l+1}\tilde{f}_l(1)^2 + \frac{1}{K^2}\bE \bE_{\geq l}\tilde{f}_l(1)^2 
\\&- \frac{2}{K^2}
\bE[(\bE_{\geq l+1}\tilde{f}_l(1)|\phi_{l-1})(\bE_{\geq l}\tilde{f}_l(1))] \\
& = \frac{2}{K^2} \bE \tilde{f}_l(1)^2 - \frac{2}{K^2} \bE[(\bE_{\geq l}\tilde{f}_l(1))^2]\\
& \leq \frac{2}{K^2} \bE \tilde{f}_l(1)^2
\end{align*}
Notice that $\tilde{f}_l(0) = 0$ and $\frac{d^2}{dt^2}\tilde{f}_l(t) \geq 0$.
Therefore,
\begin{equation*}
\tilde{f}_l'(0)  \leq \tilde f_l(1) \leq \tilde{f}_l'(1)
\end{equation*}
and
\begin{equation*}
\bE[\tilde{f}_l(1)^2]  \leq \bE[\tilde{f}_l'(0)^2] + \bE[\tilde{f}_l'(1)^2] 
\end{equation*}
This shows that our task is reduced to a proof of $\bE[\tilde{f}_l'(0)^2] = O(1)$,
$\bE[\tilde{f}_l'(1)^2] = O(1)$. 
This is a technical calculation and is given in the next lemma.
\end{proof}
\begin{lemma}\label{lem:bounddf0df1}
$\bE [(\tilde{f}_l'(0))^2] = O(1),\quad \bE [(\tilde{f}_l'(1))^2] = O(1)$
\end{lemma}
\vskip 0.25cm
\begin{proof}
From convexity, 
\begin{align*}
&(\tilde{f}'(t))^2 \leq\sum_{\mbx^0}p_{\mbX}(\mbx^0)\langle R_l^2\rangle_{\tilde H_l(t)}\\
& \leq
3\sum_{\mbx^0}p_{\mbX}(\mbx^0)\Big\langle\Big(-\frac{1}{2N}\sum_i(s_{il})^2(x^0_l-x_l)^2\Big)^2\Big\rangle _{\tilde H_l(t)}\\&+ \Big\langle\Big(\sum_{k\neq
l}\frac{1}{N}\sum_i
s_{ik}s_{il}(x^0_{l}-x_l)(x^0_{k}-x_k)\Big)^2\Big\rangle_{\tilde
H_l(t)}\\&+\Big\langle\Big(\sum_in_i\frac{1}{\sqrt N}s_{il}
x_l\Big)^2\Big\rangle_{\tilde H_l(t)}
\end{align*}
We will find a uniform bound for each term in the above sum over $\mbx^0$. Let
us consider a particular term in the above sum and set $x^0_{k} - x_k=z_{0k}$. We use the simple bound of $z_{0k}^2 \leq 4$ in the following and
hence we remove the average over $\mbx^0$.
\begin{align*}
\bE &[(\tilde{f}_l'(0))^2] \leq  12 \\&+
3\bE\Big\langle\sum_{k_1,k_2\neq l}\frac{1}{N^2}\sum_{i_1,i_2}s_{i_1k_1}s_{i_1l}s_{i_2k_2}s_{i_2l}z_{0k_1}z_{0k_2}z_{0l}^2\Big\rangle_{\tilde
H_l(0)}\\&+3\bE\Big\langle\frac{1}{N}\sum_{i_1,i_2}n_{i_1}n_{i_2}s_{i_1l}s_{i_2l}\Big\rangle_{\tilde{H}_l(0)}
\end{align*}
Since $\tilde{H}(0)$ does not depend on $s_{il}$ and since they are symmetric
random variables, in the above sums only those terms remain where $s_{il}$ are
repeated even number of times. Let $J_{kl} = \frac1N\sum_{i}s_{ik}s_{il}$ and
$\Vert J \Vert$ denote its largest singular value. Therefore,
\begin{align*}
\bE[(\tilde{f}_l(0)')^2] &\leq 12 +
3\bE\Big\langle\sum_{k_1k_2}\frac{1}{N}J_{k_1,k_2}z_{0k_1}z_{0k_2}z_{0l}^2\Big\rangle_{\tilde
H_l(0)}  +3\\
& \leq 15 + 3\times2^4\bE\Vert J\Vert + 3 = O(1)
\end{align*}
where we use that $\bE\Vert J\Vert = (1+\sqrt \beta)^2$. 
For bounding $\bE[(\tilde{f}_l(1)')^2]$ we use symmetry of the indices and take
the sum over $l$ and divide by $K$. Let $A_{ij} = \frac1K\sum_l s_{il}s_{jl}$.
\begin{align*}
\bE[(\tilde{f}_l'(1))^2] &\leq 12 + 3
\bE\Big\langle\frac1K\sum_l\sum_{k_1,k_2}J_{lk_1}J_{lk_2}z_{0k_1}z_{0k_2}z_{0l}^2\Big\rangle\\
&+3\bE\Big\langle\frac1K\frac{1}{N}\sum_{i_1,i_2}n_{i_1}n_{i_2}\sum_{l}s_{i_1l}s_{i_2l}\Big\rangle\\
& \leq 12 + 6\times2^4 \bE\Vert J\Vert^2+3\bE[\Vert
A\Vert\frac{1}{N}\sum_{i}(n_{i})^2]\\
& = 12 + 96\bE\Vert J\Vert^2 + 3 \bE\Vert A\Vert = O(1)
\end{align*}
In order to estimate $\bE\Vert J\Vert$ and $\bE\Vert A\Vert$ one can use
standard methods (see for example \cite{BoG98})
\end{proof}

\section{Estimates \eqref{eqn:indsecondderivative} and
\eqref{eqn:indthirdderivative}}\label{appen:independence}


Let $z^{(\alpha)}_k =x_{0k}-x^{(\alpha)}_k$ and $\mbz^{(\alpha)}$ denote the
vector $(z_1^{(\alpha)},\dots,z_K^{(\alpha)})$.
Let us split the contribution from $T_1 - T_2$ in to $T_{11} + T_{12}$
corresponding to the two terms appearing in \eqref{eqn:T1-T2}. 
For $T_{11}(i,k)$, we get 
\begin{align}
T_{11}(i,k) &= \frac{1}{2\sqrt t}\bE_{r_{ik}}[(r_{ik}^2-1)\int_{0}^{r_{ik}}\bE_{\sim
r_{ik}}\left[\frac{\partial^2g_{ik}(u)}{\partial u^2} du\right]]
\end{align}
where $g_{ik}(u)$ denotes the function in \eqref{eqn:gik} with $r_{ik} = u$.
Let $\langle.\rangle_{t,i,k}$ denote the Gibbs measure with $r_{ik} = u$. Let $\mbv_i^k(t)$ denote the 
vector $\mbv_i(t)$ with $r_{ik}$ replaced by $u$.
We now show that the term inside the integral decays with $N$.
\begin{align}\label{eqn:gikfirstder}
\frac{\partial g_{ik}(u)}{\partial u}= &\frac{1}{2\sigma^4 KN} \Big(\sigma^2\bE\blangle z_k^2\brangle_t- \bE
\blangle (n_i + \mbv_i(t)\cdot \mbz)^2 z_k^2\brangle_t\nonumber\\
&+\bE
\blangle( n_i + N^{-\frac12}\mbv_i(t)\cdot \mbz^{(1)})(n_i + N^{-\frac12}\mbv_i(t) \cdot\mbz^{(2)})z^{(1)}_k z^{(2)}_k\brangle_t\Big)
\end{align}

\begin{align}\label{eqn:giksecondder}
\frac{\partial^2g_{ik}(u)}{\partial u^2}& = \frac{1}{2\sigma^6 K
N^\frac32}\Big(-\sigma^2\blangle(n_i + N^{-\frac12}\mbv_i^k(t)\cdot
\mbz)3z_k^3\sqrt{t}\brangle_{t,i,k} \nonumber\\
& +3\sigma^2\blangle (n_i + N^{-\frac12}\mbv_i^k(t)\cdot \mbz^{(2)})(z_k^{(1)})^2
z_k^{(2)}\sqrt{t} \brangle_{t,i,k} \nonumber\\ 
& + \blangle (n_i + N^{-\frac12}\mbv_i^k(t)\cdot \mbz)^3z_k^3\sqrt{t} \brangle_{t,i,k} \nonumber\\
& -3\blangle (n_i + N^{-\frac12}\mbv_i^k(t)\cdot \mbz^{(1)})^2(n_i + N^{-\frac12}\mbv_i^k(t)\cdot \mbz^{(2)})(z_k^{(1)})^2z_k^{(2)}\sqrt{t}\brangle_{t,i,k} \nonumber\\
& + 2\blangle \Pi_{a = 1,2,3}(n_i + N^{-\frac12}\mbv_i^k(t)\cdot \mbz^{(a)})z_k^{(1)}z_k^{(2)}z_k^{(3)}\sqrt{t} \brangle_{t,i,k}\Big)
\end{align}
The Hamiltonians corresponding to $\langle.\rangle_t$ and $\langle.\rangle_{t,i,k}$ are 
\begin{align*}
H(\mbz) = -\frac{1}{2\sigma^2}\Vert\mbn + N^{-\frac12} \mtv(t)\mbz\Vert^2
,\;\;\;H_{i,k}(\mbz) = -\frac{1}{2\sigma^2}\Vert\mbn + N^{-\frac12} \mtv_{i,k}(t)\mbz\Vert^2
\end{align*}
where $\mtv_{i,k}(t)$ differs from $\mtv(t)$ only in the $(i,k)$th entry with $u$ replacing $r_{ik}$.
Expanding $H_{i,k}$, 
\begin{align*}
H_{i,k}(\mbz) =& -\sum_{j\neq i}(n_j + N^{-\frac12}\mbv_j \cdot \mbz)^2 - (n_i +
N^{-\frac12}\sum_{l\neq k}v_{il}z_l + N^{-\frac12}\sqrt{1-t}s_{ik}z_k)^2\\
& -\frac{u^2tz_{k}^2}{N} + \frac{u\sqrt{t}z_{k}}{\sqrt N} (n_i +
N^{-\frac12}\sum_{l\neq k}v_{il}z_l+ N^{-\frac12}\sqrt{1-t}s_{ik}z_k)
\end{align*}
Let the sum of the first two terms be denoted as $H'_{ik}(\mbz)$ and the terms involving $u$ be $H''_{ik}(\mbz)$. Consider the following set
\begin{align*}
G = \{\mbn,\mtr,\mts: \forall i \;\;\frac{1}{\sqrt N} \vert n_i \vert + \frac{1}{N}
\sum_k2\vert r_{ik} \vert + \frac{1}{N}
\sum_k2\vert s_{ik} \vert \leq C \}
\end{align*}
For sufficiently large $C$ we have $P(G^{c}) = O(e^{-\alpha N})$ for some
constant $\alpha > 0$. If $(\mbn,\mts,\mtr) \in G$, then for all $\mbz\in\{0,2\}^K$
\begin{align}
\vert H''_{i,k}(\mbz)\vert \leq \frac{4|u|^2}{N} + 2|u| C\equiv C'(u).
\end{align}
Therefore for the first term in the equation \eqref{eqn:giksecondder}
\begin{align*}
&\Big|\bE_{\sim r_{ik}}\blangle(n_i + N^{-\frac12} \mbv^k_i\cdot\mbz)\brangle_{t,i,k}\Big| \\ 
&\leq \bE\blangle \frac{\sum_{\mbz}
e^{-H'_{ik}(\mbz)}e^{\frac{C'(u)}{2\sigma^2}}|n_i + N^{-\frac12}
\mbv^k_i\cdot\mbz -u\sqrt {t}
N^{-\frac12}z_k|}{\sum_{\mbz}e^{-H'_{ik}(\mbz)}e^{-\frac{C'(u)}{2\sigma^2}}}\indicator{G}\brangle_{t,i,k}
+ O\Big(\frac{|u|}{\sqrt N}\Big)\\
& + \bE\blangle|n_i + N^{-\frac12} \mbv_i\cdot\mbz | \indicator{G^c} \brangle_{t,i,k}
\end{align*}
The expectation over $G^c$ can be bounded as $O(e^{-\alpha N})O(|u|)$. Therefore
the last two terms contribute $O(\frac{|u|}{\sqrt N})$. For the first term after
we have removed the terms with $u$ dependence, the Hamiltonian $H'_{ik}$
satisfies Nishimori symmetry. Therefore we get the first term to be equal to, 
\begin{align*}
&\bE_{\mts,\mtr}\int \frac{1}{2^K}e^{2\frac{C'(u)}{2\sigma^2}}{\sum_{\mbz}
e^{-H_{ik}(\mbz)}|n_i + N^{-\frac12} \mbv^k_i\cdot\mbz -u\sqrt {t}
N^{-\frac12}z_k|}\;d\mbn \nonumber \\ 
& =  \sqrt{\frac{\sigma^2}{2\pi}}e^{\frac{C'(u)}{\sigma^2}} 
\end{align*}
Note that the above integral is a Gaussian integral and can be evaluated easily.
Using similar method, we can show that 
\begin{align}
\bE_{\sim r_{ik}}\Big[\frac{\partial^2 g_{ik}(u)}{\partial u^2}\Big] \leq
O(1)e^{\frac{3C'(u)}{\sigma^2}} + O(N^{-\frac12})|u|^3
\end{align}
The exponent $3$ is due the occurrence of $3$ replicas in the equation
\eqref{eqn:giksecondder}. Therefore,
\begin{align}
\bE_{r_{ik}}&\Big[(r_{ik}^2-1)\int_{0}^{r_{ik}}\bE_{\sim
r_{ik}}\Big[\frac{\partial^2 g_{ik}(u)}{\partial u^2}\Big] du\Big]\nonumber\\
\leq &\bE_{r_{ik}}\Big[r_{ik}^2\int_{0}^{r_{ik}}
N^{-\frac52}(O(1)e^{3\frac{C'(u)}{\sigma^2}} +
O(N^{-\frac12}|u|^3))du\Big]\nonumber\\
\leq &  O(N^{-\frac52})  
\end{align}
where we have used the assumption A for the distribution of $r_{ik}$. Now summing this over all $i,k$ we get 
\begin{align}
|T_{11}|\leq O(N^{-\frac12})
\end{align}
Now consider the term $T_{13}$. For this we have to evaluate the following term. 
\begin{align*}
\frac{\partial^3 g_{ik}(u)}{\partial u^3} &= \frac{t}{2\sigma^8
KN^2}\Big(-3\sigma^4\bE\langle z_k^4\rangle_t + 6\sigma^2 \bE\blangle(n_i + N^{-\frac12} \mbv_i(t) \cdot \mbz)^2 z_k^4 \brangle_t \\
& -12 \sigma^2 \bE\blangle\Pi_{a=1,2}(n_i + N^{-\frac12} \mbv_i(t) \cdot \mbz^{(a)}) (z_k^{(1)})^3z_k^{(2)}\brangle_t\\
& +3\sigma^4 \bE\blangle (z_k^{(1)})^2(z_k^{(2)})^2 \brangle_t - 6\sigma^2
\bE\blangle(n_i + N^{-\frac12} \mbv_i(t) \cdot \mbz^{(2)})^2
(z_k^{(1)})^2(z_k^{(2)})^2\brangle_t\\
& + 9\sigma^2 \bE\blangle\Pi_{a=2,3}(n_i + N^{-\frac12} \mbv_i(t) \cdot \mbz^{(a)})
(z_k^{(1)})^2z_k^{(2)}z_{k}^{(3)}\brangle_t\\
& - \bE\blangle(n_i + N^{-\frac12} \mbv_i(t) \cdot \mbz)^4z_k^4\brangle_t\\
& +4 \bE\blangle(n_i + N^{-\frac12} \mbv_i(t) \cdot \mbz^{(1)})^3(n_i + N^{-\frac12} \mbv_i(t) \cdot \mbz^{(2)})
(z_k^{(1)})^3z_k^{(2)}\brangle_t\\
& +  3 \bE\blangle \Pi_{a=1,2}(n_i + N^{-\frac12} \mbv_i(t) \cdot \mbz^{(a)})^2(z_k^{(1)})^2(z_k^{(2)})^2 \brangle_t \\
& - 12 \bE\blangle (n_i + N^{-\frac12} \mbv_i(t) \cdot \mbz^{(1)})^2(n_i + N^{-\frac12} \mbv_i(t) \cdot \mbz^{(2)}) 
 (n_i + N^{-\frac12} \mbv_i(t) \cdot \mbz^{(3)})(z_k^{(1)})^2z_k^{(2)}z_k^{(3)} \brangle_t\\
& + 6 \bE\blangle \Pi_{a = 1,2,3,4}(n_i + N^{-\frac12} \mbv_i(t) \cdot \mbz^{(a)})z_k^{(1)}z_k^{(2)}z_k^{(3)}z_k^{(4)} \brangle_t\Big)
\end{align*}
We can prove along similar lines that $|T_{12}| \leq O(N^{-1})$.

\section{Nishimori Identities\label{apen:nishimori}}
{\em Proof of Lemma \ref{lem:nishimorimq}.}
We only give a brief sketch because the method is standard (see for example
\cite{Nish01,Nis02}). 
One writes fully explicitly the expression for 
$\bP^t_{\magn}(x)$ and performs the gauge transformation
$x_k\to x_k^0 x_k$, $s_{ik}\to x_k^{0} s_{ik}$ where $\mbx^0$ is an arbitrary binary sequence. 
Since $\bP^t_{\magn}(x)$ does not depend 
on $\mbx^0$ we sum over all such $2^K$
sequences and obtain a lengthy expression.
Exactly the same procedure is applied to $\bP^t_{q_{12}}(x)$ and one gets another lengthy expression. Then one can recognize 
that these two expressions are the same.
\\
\\
{\em Proof of Lemma \ref{lem:nishimori2}.}

{\em Proof of \eqref{eqn:X11}.}
We will prove it for $t=1$ and for general $t$ it is similar. Let the
transmitted sequence be the all one sequence, and the received vector be $\mbr =
\sigma \mbn + \sqrt{\frac{1}{N}} \mts \mbone$ where $n_i\sim\mathcal{N}(0,1)$. The proof follows by using gauge symmetry. Let
$\mbu$ denote the $K$ dimensional vector $(u,\dots,u)$. 
\begin{align}
\bE[\langle\Vert \mcZ\Vert^2 \rangle_{1,u}]
& = \bE_{\mtS}\Big[\int\frac{1}{(2\pi u)^{\frac{K}{2}}(2\pi\sigma^2)^{\frac{N}{2}}}
e^{-\frac{\Vert\mbh-\mbu\Vert^2}{2u}}e^{-\frac1{2\sigma^2} \Vert\mbr -
N^{-\frac12}\mts\Vert^2}
\langle\Vert\mcZ\Vert^2\rangle_{1,u} d\mbr\; d\mbh\Big]\nonumber\\ 
& = \bE_{\mtS}\Big[\int \frac{1}{(2\pi u)^{\frac{K}{2}}(2\pi \sigma^2)^{\frac{N}{2}}}e^{-\frac{\Vert\mbh\Vert^2}{2u} + \mbh\cdot\mbone
-\frac{Ku}{2}}e^{-\frac1{2\sigma^2} \Vert\mbr - N^{-\frac12}\mts\Vert^2}
\frac{\sum_{\mbx}e^{-\frac1{2\sigma^2} \Vert\mbr -
N^{-\frac12}\mts\mbx\Vert^2+\mbh\cdot\mbx}\Vert\mcZ\Vert^2}{\sum_{\mbx}e^{-\frac1{2\sigma^2}
\Vert\mbr - N^{-\frac12}\mts\mbx\Vert^2+\mbh\cdot\mbx}} d\mbr\;d\mbh\Big]\nonumber\\ 
&= \frac{1}{2^{K}}\bE_{\mtS}\Big[\int \frac{1}{(2\pi u)^{\frac{K}{2}}(2\pi
\sigma^2)^{\frac{N}{2}}}e^{-\frac{\Vert\mbh\Vert^2}{2u}
-\frac{Ku}{2}}\sum_{\mbx^0}e^{-\frac1{2\sigma^2} \Vert\mbr -
N^{-\frac12}\mts\mbx^0\Vert^2+ \mbh\cdot\mbx^0 }\nonumber\\
&\quad\quad\quad\quad\quad\quad \frac{\sum_{\mbx}e^{-\frac1{2\sigma^2} \Vert\mbr -
N^{-\frac12}\mts\mbx\Vert^2+\mbh\cdot\mbx}\Vert\mcZ\Vert^2}{\sum_{\mbx}e^{-\frac1{2\sigma^2}
\Vert\mbr - N^{-\frac12}\mts\mbx\Vert^2+\mbh\cdot\mbx}} \; d\mbr\; \;
d\mbh\Big]\label{eqn:gauge1}\\ 
& =N\nonumber
\end{align}
\eqref{eqn:gauge1} is obtained by performing the gauge transformation $x_k\to x_kx^0_{k}$, $s_{ik}\to s_{ik}x^0_{k}$
and $h_{k} \to h_k x^0_{k}$ and summing over all the $2^K$ possibilities of
$\mbx^0$. Now canceling the summation over $\mbx^0$ with the denominator and
then integrating we get it to be equal to $N$.

{\em Proof of \eqref{eqn:X12}.}
The proof is complete if we show $\bE[\langle (\mbn \cdot \mcZ^{(2)})(\mbx^{(1)}\cdot\mbz^{(2)}) \rangle_{t,u}] = 0$.
We will prove this for $t=1$ and it is similar for other $t$.
\begin{align*}
\bE[\langle& (\mbn \cdot \mcZ^{(2)})(\mbx^{(1)} \cdot \mbz^{(2)}) \rangle_{t,u}] \\
&= \sum_{i,k}\bE[\langle(r_i -
 N^{-\frac12} \sum_l s_{il})(r_i - N^{-\frac12} \sum_l s_{il}x_l^{(2)})(x_k^{(1)}-x_k^{(1)}x_k^{(2)})\rangle_{t,u}]
\end{align*}
Now performing the gauge transformation $x_{k}^{(1)}\to x_{k}^{(1)}x_k^0$, $x_{k}^{(2)}\to x_{k}^{(2)}x_k^0$, $s_{ik}\to s_{ik}x_k^0$ and $h_k\to h_k x^0_k$ we get 
\begin{align*}
\sum_{i,k}\bE[\langle(r_i -
 N^{-\frac12} \sum_l s_{il}x^0_l)(r_i - N^{-\frac12} \sum_l s_{il}x_l^{(2)})(x_k^{(1)}x^0_k-x_k^{(1)}x_k^{(2)})\rangle_{t,u}]
\end{align*}
This quantity can be shown to be equal to $0$ by noticing that the $\mbx^0$ and $\mbx^{(2)}$ play symmetric roles.

\section{Proof of inequality (\ref{eqn:centrallimbound})
\label{apen:centrallimbound}}
For a given configuration of $\mbz$, $\frac{1}{\sqrt N}\sum_l s_{il} z_l \equiv Z_i$ is a
Gaussian random variable with mean $0$ and variance smaller than $4$.
Thus for $n_i\sim \mathcal{N}(0,1)$ and independent of $Z$,
\begin{align*}
\bE [e^{\frac{n_iZ_i}{\alpha}}] = \bE [e^{\frac{-n_iZ_i}{\alpha}}] \leq \sqrt\frac{\alpha^2}{\alpha^2 - 4}
\end{align*}
If $\alpha > 2$, we have both the expectations to be less than some constant $C
> 1$. Therefore for any $\mbz$
\begin{align*}
\bE [e^{-\frac{1}{\alpha}{N^{-\frac12}}\sum_{i,l} n_i s_{il} z_l}] =\bE[
e^{\frac{1}{\alpha}N^{-\frac12}\sum_{i,l} n_i
s_{il} z_l} ]\leq C^N
\end{align*}
Using the Markov inequality, 
\begin{align*}
\bP\Big(\Big\vert\frac1{\alpha} \sum_i n_i\frac{1}{\sqrt N}\sum_k s_{ik} z_k
\Big\vert > yN\Big)
\leq 2 C^Ne^{-yN}
\end{align*}
Using the union bound over $\mbz$, for $y$ large enough there exists a constant $\gamma > 0$ such that 
\begin{align*}
\bP\Big(\exists \mbz \in \{0,2\}^K: \Big\vert\frac1{N^{3/2}} \sum_{i,k} n_is_{ik}
z_k\Big\vert >\alpha y \Big) \leq   2^{-\gamma N}
\end{align*}
Let $G$ be the event that $\Big\vert\frac1{N^{3/2}} \sum_{i,k} n_is_{ik}
z_k\Big\vert >\alpha y$ holds for all $\mbz$. Splitting the expectation into two parts corresponding to $G$ and $G^c$ and using Cauchy-Schwartz inequality, we have
\begin{align*}
&\bE\Big\langle\Big(\frac1{N^{3/2}}\sum_{i,k} n_i
s_{ik} z_k \Big)^2\Big\rangle_t\\
& \leq \alpha^2 y^2 + \sqrt{P(G^c)} \Big(\bE\Big\langle\Big(\frac1{N^{3/2}} \sum_{i,l}
n_is_{ik} z_k \Big)^4\Big\rangle_t\Big)^{1/2}\\
& \leq \alpha^2 y^2 + O( 2^{-\frac\gamma2 N})
\end{align*}
\vskip 0.25cm

\end{appendix}

\section*{Acknowledgments}
We would like to thank Shrinivas Kudekar, Olivier L\'ev\^eque, Andrea Montanari, and R\"udiger Urbanke
for useful discussions. The work presented in this paper is partially supported by the
National Competence Center in Research on Mobile Information and
Communication Systems (NCCR-MICS), a center supported by the Swiss
National Science Foundation under grant number 5005-67322.


\end{document}